\documentclass[12pt,american,english]{article}
\usepackage[T1]{fontenc}
\usepackage[latin9]{inputenc}
\usepackage{amsmath}
\usepackage{amsthm}
\usepackage{amssymb}
\usepackage{geometry}
\usepackage{setspace}
\usepackage[authoryear]{natbib}
\makeatletter
\theoremstyle{definition}
\newtheorem{defn}{\protect\definitionname}
\theoremstyle{plain}
\newtheorem{lem}{\protect\lemmaname}
\theoremstyle{plain}
\newtheorem{prop}{\protect\propositionname}
\theoremstyle{plain}
\newtheorem{thm}{\protect\theoremname}
\theoremstyle{plain}
\newtheorem{assumption}{\protect\assumptionname}
\theoremstyle{plain}
\newtheorem{cor}{\protect\corollaryname}
\theoremstyle{definition}
 \newtheorem{example}{\protect\examplename}

\@ifundefined{date}{}{\date{}}

\usepackage{nicefrac}
\usepackage{setspace}
\usepackage{chngcntr}
\usepackage{microtype}
\usepackage{lmodern}
\usepackage{tikz}
\usepackage{tkz-tab}
\usepackage{caption}
\usepackage{subcaption}
\usepackage{graphicx}

\usepackage{abstract}

\usepackage{parskip}

\title{Data-Driven Persuasion}
\usepackage{amsmath,amsthm,amssymb,amsfonts} 

\author{Maxwell Rosenthal\thanks{Georgia Institute of Technology.  Email: \href{mailto:rosenthal@gatech.edu}{rosenthal@gatech.edu.}}}

\date{July 14, 2025}

\usepackage[colorlinks=true,linkcolor=blue,citecolor=blue,urlcolor=blue,hyperfootnotes=false]{hyperref}

\makeatother

\usepackage{babel}
\addto\captionsamerican{\renewcommand{\assumptionname}{Assumption}}
\addto\captionsamerican{\renewcommand{\corollaryname}{Corollary}}
\addto\captionsamerican{\renewcommand{\definitionname}{Definition}}
\addto\captionsamerican{\renewcommand{\examplename}{Example}}
\addto\captionsamerican{\renewcommand{\lemmaname}{Lemma}}
\addto\captionsamerican{\renewcommand{\propositionname}{Proposition}}
\addto\captionsamerican{\renewcommand{\theoremname}{Theorem}}
\addto\captionsenglish{\renewcommand{\assumptionname}{Assumption}}
\addto\captionsenglish{\renewcommand{\corollaryname}{Corollary}}
\addto\captionsenglish{\renewcommand{\definitionname}{Definition}}
\addto\captionsenglish{\renewcommand{\examplename}{Example}}
\addto\captionsenglish{\renewcommand{\lemmaname}{Lemma}}
\addto\captionsenglish{\renewcommand{\propositionname}{Proposition}}
\addto\captionsenglish{\renewcommand{\theoremname}{Theorem}}
\providecommand{\assumptionname}{Assumption}
\providecommand{\corollaryname}{Corollary}
\providecommand{\definitionname}{Definition}
\providecommand{\examplename}{Example}
\providecommand{\lemmaname}{Lemma}
\providecommand{\propositionname}{Proposition}
\providecommand{\theoremname}{Theorem}

\begin{document}
\selectlanguage{american}%
\maketitle

\begin{abstract}
This paper develops a data-driven approach to Bayesian persuasion.  The receiver is privately informed about the prior distribution of the state of the world, the sender knows the receiver's preferences but does not know the distribution of the state variable, and the sender's payoffs depend on the receiver's action but not on the state.
Prior to interacting with the receiver, the sender observes the distribution of actions taken by a population of decision makers who share the receiver's preferences in best response to an unobserved distribution of messages generated by an unknown and potentially heterogeneous signal.  The sender views any prior that rationalizes this data as plausible and seeks a signal that maximizes her worst-case payoff against the set of all such distributions.  We show positively that the two-state many-action problem has a saddle point and negatively that the two-action many-state problem does not. In the former case, we identify adversarial priors and optimal signals. In the latter, we characterize the set of robustly optimal Blackwell experiments.

\end{abstract}


\selectlanguage{english}%

\section{Introduction}

A prosecutor wishes to secure a conviction, the defendant is either
innocent or guilty, and the judge is indifferent between type one
and type two errors. As in the original formulation of both the example
and the general persuasion problem (\citet*{kamenica2011bayesian}),
the judge is experienced and knows the prior probability of guilt.
Unlike in the standard formulation, the prosecutor is inexperienced
and does not. Instead, she knows only that the court convicts $40\%$
of defendants.

What can the prosecutor infer about the unconditional probability
of guilt? At one extreme, each of the court's convictions might have
been won by a zealous district attorney who acted to maximize the
unconditional probability of conviction. In that case, because half
of the convicted defendants and none of the acquitted defendants were
guilty, the prior probability of guilt is $20\%$. At the other extreme,
each conviction might be due to a reluctant district attorney who
sought acquittals for as many defendants as possible. In that case,
all of the convicted defendants and half of the acquitted defendants
were guilty, and the prior probability of guilt is $70$\%.

What should the prosecutor do with this information? In light of her
preference for convictions, the robustly optimal communication strategy
is to recommend the court convict every guilty defendant and $25\%$
of innocent defendants, as if she were certain that each previous
case was brought by a zealous district attorney. This signal is guaranteed
to recover a $40\%$ conviction rate and will yield convictions more
frequently if a portion of the court's earlier decisions were in response
to fair-minded prosecutions.

This paper proposes a model of robust data-driven persuasion in which
the receiver knows the distribution of the state of the world but
the sender does not. There are finitely many actions, finitely many
states, and the sender's payoff depends on the former but not on the
latter. Prior to interacting with the receiver, the sender observes
the distribution of actions taken by a population of decision makers
in best response to a distribution of messages generated by an unknown
signal. Both the messages generated by these signals and the signals
themselves are unobserved and heterogeneous, the decision makers'
preferences are known and homogeneous, and the receiver is a representative
member of the population. With full control over what is communicated
to the receiver, the sender seeks a signal that maximizes her worst
case payoff against the set of prior distributions that rationalize
the action data.

Our first task is to identify the set of priors that are consistent
with the sender's observations. In the language of the motivating
example, our notion of rationalization is general enough to allow
for situations in which the court sometimes hears zealous arguments,
sometimes hears reluctant arguments, and sometimes hears arguments
that are neither; the court fields recommendations that are more varied
than simply ``acquit'' or ``convict''; and an on-the-fence court
convicts some defendants and acquits others. While this flexibility
might in principle complicate our analysis, we show to the contrary
that a given prior rationalizes the action data if and only if it
does so via a fixed signal in which each of the district attorneys
in the prosecutor's data recommended convictions and acquittals with
homogeneous frequencies and in an incentive compatible manner. With
that reduction in hand, we use the supporting hyperplane theorem to
characterize the set of rationalizing posteriors and Bayes plausibility
(\citet{kamenica2011bayesian}) to combine rationalizing posteriors
into rationalizing priors.

We initiate our analysis of the sender's problem itself with a formal
analysis of the two-action two-state benchmark problem. As in the
motivating example, the lowest\footnote{In the two-action two-state problem, we call the state in which the
receiver prefers the sender's disfavored action the \emph{low state
}and the state in which he prefers the sender's favored action the
\emph{high state}. We order priors by the probability they assign
to the latter.} prior that rationalizes the data and the optimal signal for the counterfactual
problem in which the sender knows the state of the world is distributed
according to that prior form a saddle point of the sender's problem.
Because improvements to the prior probability of the high state increase
the frequency with which the signal recommends the sender's preferred
action and also relax the incentive compatibility constraint associated
with that recommendation, their effect is to unambiguously improve
the sender's payoff. The worst-case problem is as if the sender knew
the state of the world were distributed according to the lowest rationalizing
prior.

Our positive result for the two-state two-action problem problem extends
readily to the two-state many-action problem as long as the sender's
payoffs are increasing in the receiver's action.\footnote{In the two-state problem, we label the states \emph{low }and \emph{high
}arbitrarily and enumerate the receiver's actions in ascending order
according to their payoff in the high state. The substance of our
monotonicity assumption lies in requiring that the sender's utility
is increasing in one of the two potential labelings of the receiver's
actions.} The lowest prior in the identified set is adversarial; the signal
designed for that prior is robustly optimal; and improvements to the
prior probability of the high state (i) strictly increase the probability
with which the optimal signal sends the higher of its two messages
and (ii) weakly increase the receiver's best response to both of those
messages. By monotonicity of the sender's preferences, increasing
the prior improves the sender's payoff.

In contrast to the two-state many-action problem, our positive result
for the two-by-two problem does not extend to the two-action many-state
problem under any non-trivial specification of sender and receiver
preferences. Unlike in the two-state problem, the set of rationalizing
priors is multidimensional and its lower boundary contains a multiplicity
of distributions that are not ordered in any relevant sense. While
the optimal signal in the known-prior counterfactual recovers the
action distribution observed in the data for each one of the boundary
priors, any signal that is optimal in one of those counterfactuals
is suboptimal in all of the others. No signal is guaranteed to perform
as well as the unobserved signal that generated the data and no prior
is adversarial.

The non-existence of an adversarial prior in the two-action many-state
problem implies that exact knowledge of the true distribution of the
state variable is of guaranteed value to the sender. In fact, less
information suffices. Formally, we (i) call a Blackwell experiment
\emph{reliable }if its output is guaranteed to induce a saddle point
in the sender's problem and (ii) completely characterize the set of
reliable experiments as those that are \emph{ordered }in terms of
the one-dimensional structure that their outcomes impose on the set
of rationalizing priors. We interpret reliability as a robust optimality
criterion and provide a formal justification for that interpretation
in the body; conversely, ordered experiments are those that satisfy
a pair of restrictions to the null space of the stochastic matrix
that represents the experiment.

Blackwell experiments play a different role in our environment than
they do in their traditional setting, wherein the decision maker receives
a single message and uses their knowledge of its conditional distribution
to update their prior beliefs about the state of the world. In our
model, the sender instead observes the entire distribution of messages
and uses her knowledge of their conditional distribution to infer
restrictions on the prior distribution of the state variable. Within
that framework, the simplest class of experiments that satisfy our
order criterion are those that pin down the probabilities associated
with each of the states in which the receiver strictly prefers the
sender's disfavored action. More generally, while we do characterize
the reliability of \emph{simple }experiments that (i) partition the
state space into cells and (ii) reveal the probabilities of each cell
in the partition, reliable experiments need not be simple.

The paper is organized as follows. In order to fix ideas, we lay out
the model in Section \ref{Section: model} before reviewing the literature
in Section \ref{Section: literature review}. In Section \ref{SECTION: identification}
we study the identification of the sender's problem. Next, in Sections
\ref{Section: two by two}--\ref{Section: two actions and many states}
we analyze the two-action two-state benchmark problem, the two-state
many-action problem, and the two-action many-state problem, respectively.
We conclude in Section \ref{Section: conclusions} and present technical
material in the Appendix that follows.

\section{\label{Section: model}Model}

Given a finite set $X$, we write $\Delta(X)$ for the set of probability
distributions on $X$ and $\delta(x)$ for the degenerate distribution
at $x$. We interpret elements of $\Delta(X)$ variously as measures
or as real vectors, give the real numbers and their products the Euclidean
topology, and write $\varepsilon$ to indicate an arbitrarily small
strictly positive real number. Finally, given a matrix $M_{m\times n}$,
we write $M_{i}\equiv(M_{1i},...,M_{mi})$ for its columns.

\subsection{Action, states, and payoffs}

Actions are elements $a$ of finite set $A$, states are elements
$\omega$ of finite set $\Omega$, and there are at least two actions
and at least two states. The receiver's payoffs are described by utility
function $u:A\times\Omega\to\mathbb{R}$, the sender's payoffs are
described by state-independent and injective utility function $v:A\to\mathbb{R}$,
and the state of the world is distributed according to prior distribution
$\mu\in\Delta(\Omega)$. The map $\phi:A\to\mathbb{R}^{\vert\Omega\vert}$
defined by $a\mapsto(u(a\vert\omega))_{\omega\in\Omega}$ is injective.

\subsection{Signals}

A \emph{signal} is a pair $(\pi,S)$, where $S$ is a finite set of
\emph{messages} and the map $\pi:\Omega\to\Delta(S)$ assigns a distribution
of messages $\pi(\cdot\vert\omega)$ to each state $\omega$. Given
a prior $\mu\in\Delta(\Omega)$ and a signal $(\pi,S)$, we write
\begin{align*}
\mu^{s}(\omega) & \equiv\frac{\pi(s\vert\omega)\mu(\omega)}{\sum_{\omega'}\pi(s\vert\omega')\mu(\omega')}
\end{align*}
for the posterior induced by message $s$ given prior $\mu\in\Delta(\Omega)$.
Given posterior $\mu^{s}\in\Delta(\Omega)$, we write
\begin{align*}
a^{*}(\mu^{s}) & \equiv\underset{a\in A}{\arg\max}\;\sum_{\omega}u(a\vert\omega)\mu^{s}(\omega), & \widehat{a}(\mu^{s}) & \equiv\underset{a\in a^{*}(\mu^{s})}{\arg\max}\;v(a)
\end{align*}
for the set of receiver-optimal actions and the sender's preferred
element of that set, respectively, noting that $\widehat{a}$ is single-valued
per our hypothesis that the sender's utility $v$ is injective. Ties
are broken in favor of the sender and her expected utility from signal
$(\pi,S)$ and prior $\mu$ is thus
\[
V((\pi,S)\vert\mu)\equiv\sum_{s}\sum_{\omega}v(\widehat{a}(\mu^{s}))\pi(s\vert\omega)\mu(\omega).
\]
If $(\pi,S)$ is a \emph{direct }signal with $S\subset A$, we sometimes
interpret the condition $a\in a^{*}(\mu^{a})$ as an \emph{incentive
compatibility }constraint. More generally, in places where the identity
of the signal $(\pi,S)$ should be understood, we write $S^{a}(\mu)\equiv\{s\in S\vert a=\widehat{a}(\mu^{s})\}$
for the set of messages under which the receiver takes action $a$
given prior $\mu$. 

\subsection{Rationalizing priors}

The prior distribution of the state of the world is known to the receiver
but not to the sender. Instead, the sender --- who does know the
receiver's utility $u$ --- observes the distribution of best responses
by a population of receivers to some heterogeneous and unobserved
signal. She entertains the possibility that the prior might be any
that rationalize the data.
\begin{defn}
Prior $\mu_{}$ \emph{homogeneously rationalizes }action distribution
$\alpha\in\Delta(A)$ if there exists a signal $(\pi,S)$ and a choice
rule $C:S\to\Delta(A)$ such that (i) $a\in a^{*}(\mu^{s})$ for all
$a,s$ with $a\in\text{supp}(C(s))$ and (ii) $\sum_{\omega}\sum_{s}C(a\vert s)\pi(s\vert\omega)\mu(\omega)=\alpha(a)$
for all $a\in\text{supp}(\alpha)$.
\end{defn}
Prior $\mu$ homogeneously rationalizes action distribution $\alpha$
if there exists a fixed signal $(\pi,S)$ and a tiebreaking rule $C$
such that $\alpha$ is the distribution of best responses to $(\pi,S)$
when the state of the world is distributed according to $\mu$ and
ties are broken according to $C$.
\begin{defn}
Prior $\mu$ \emph{rationalizes }action distribution $\alpha\in\Delta(A)$
if there exists a finite index $g=1,...,G$, a distribution $\gamma\in\Delta(\{1,...,G\})$,
and a collection of action distributions $\alpha^{1},...,\alpha^{G}\in\Delta(A)$
such that (i) $\mu$ homogeneously rationalizes $\alpha^{g}$ for
each $g$ and (ii) $\alpha=\sum_{g}\gamma^{g}\alpha^{g}$.
\end{defn}
More generally, $\mu$ rationalizes $\alpha$ if there exists a distribution
$\gamma$ of signals $(\pi^{g},S^{g})$ such that $\alpha$ is the
average distribution of best responses to $(\pi^{g},S^{g})$ when
the state of the world is distributed according to $\mu$. We write
$\mathcal{P}(\alpha)$ for the set of all such priors.

\subsection{The sender's problem}

The \emph{sender's problem }
\[
\max_{(\pi,S)}\;\min_{\mu\in\mathcal{P}(\alpha)}\;V((\pi,S)\vert\mu)
\]
is to maximize her \emph{guarantee }$\min_{\mu\in\mathcal{P}(\alpha)}\;V((\pi,S)\vert\mu)$
against the set of priors that rationalize the data $\alpha$. Our
approach to analyzing this problem is to identify \emph{saddle points
}$((\pi^{*},S^{*}),\mu^{*})$ defined by the inequalities
\begin{align*}
\forall(\pi,S)\;V((\pi^{*},S^{*})\vert\mu^{*}) & \geq V((\pi,S)\vert\mu^{*}), & \forall\mu\in\mathcal{P}(\alpha)\;V((\pi^{*},S^{*})\vert\mu^ {}) & \geq V(\pi^{*},S^{*})\vert\mu^{*}).
\end{align*}
If the sender's problem has saddle point $((\pi^{*},S^{*}),\mu^{*})$
then $(\pi^{*},S^{*})$ is \emph{robustly optimal} and its guarantee
is the same as in the counterfactual environment in which the state
is known to be distributed according to the \emph{adversarial} prior
$\mu^{*}$. Conversely, if the sender's problem does not have a saddle
point then the value of the sender's problem is lower than the value
of the counterfactual problem in which the sender knows that the prior
is $\mu$, for any $\mu$ that rationalizes the data. Accordingly,
knowledge of the distribution of the state of the world is of guaranteed
value to the sender if and only if her problem lacks a saddle point.

\section{\label{Section: literature review}Literature review}

This paper contributes to a number of partially overlapping segments
of the extensive literature on persuasion and information design. 

\paragraph{Robust persuasion}

We share our emphasis on uncertainty and worst-case solution concepts
in persuasion problems with a handful of recent studies. \citet*{dworczak2022robust}
study an environment in which there are common prior beliefs about
the distribution of the state of the world and the sender is concerned
about the presence of unknown background signals. In contrast, our
sender has complete control over what is communicated to the receiver.

Elsewhere, \citet{kosterina2022persuasion} studies a problem in which
the sender and the receiver have different beliefs about the distribution
of the state of the world and the sender is uncertain about the receiver's
beliefs. We differ in two important ways. First, while the uncertainty
set in their model is abstract, the uncertainty set in our model it
is derived from choice data. Second, while the sender in both models
is uncertain about the receiver's beliefs, the sender in our model
believes the receiver knows the true distribution of the state of
the world. Accordingly, while our sender is uncertain about both the
distribution of messages sent to the receiver and the receiver's response
to those messages, their sender is only uncertain about the latter.

Less closely related are two studies in which a regret-minimizing
sender is unaware of the receiver's preferences (\citet*{castiglioni2020online,babichenko2022regret})
and another in which a regret-minimizing sender gradually learns about
the state variable via direct observation of its realizations (\citet*{zu2024learning}).

\paragraph{Counterfactuals in games with unknown information structures}

A number of recent studies conduct counterfactual analyses in environments
with unobserved information structures. First, \citet*{bergemann2022counterfactuals}
consider a setting in which an analyst observes the behavior of a
group of agents playing a game but does not observe either the distribution
of the state variable or the information held by the decision makers.
Aside from our focus on individual rather than group behavior and
our use of worst-case rather than Bayesian solution concepts, in our
model the sender again has full control over the receiver's information
while in theirs the information structure in the observed environment
persists in the counterfactual.

Second, \citet*{syrgkanis2017inference} study counterfactual predictions
in auctions with unknown information structures. Aside from differences
in setting, we focus our counterfactual analysis on worst-case optimal
mechanisms, while the authors take their counterfactual mechanisms
as given and bound their performance. 

Third, \citet*{magnolfi2023discrete}, develop methods for estimation
and counterfactual analysis in games with weak assumptions about information
and apply those methods to a study of real-world supermarket entry.

\paragraph{Inference about unknown information in decision problems}

Our analysis of the identification of the sender's model is related
to a pair of recent papers that study the identification of information
in economic settings. First, \citet*{doval2025revealed} study the
identification of a single-agent decision problem in a setting much
like our own. Identification is a preliminary component of our analysis,
rather than its central focus, and our results on that matter are
included primarily for the sake of completion.

Second, \citet*{libgober2025wisdom} studies an environment in which
the analyst knows the posterior beliefs of a group of agents but does
not know either the prior distribution of the state world nor the
information structure that induced the update from prior to posterior.
Because inference about those two objects is the primary objective
of that paper, we are differentiated not only by the number of agents
but also most importantly by our emphasis on the design of robustly
optimal counterfactual information structures.

\section{\label{SECTION: identification}Identification}

Before proceeding to our analysis of the sender's problem, it will
be useful to develop an understanding of the identification problem
that precedes it. As we have noted in the introduction, our definition
of rationalization allows for the action data to be generated in response
to heterogeneous signals, signals that send more messages than actions,
and complex tiebreaking rules. Our first step is to develop an equivalence
with a simpler definition.
\begin{lem}
\label{Lemma: straightforward rat}The following three statements
are equivalent:
\begin{enumerate}
\item[(i)] $\mu$ rationalizes $\alpha$;
\item[(ii)] $\mu$ homogeneously rationalizes $\alpha$; 
\item[(iii)] there exists a signal $(\pi,S)$ such that $S\subset A$, $a\in a^{*}(\mu^{a})$
for all $a\in S$, and $\sum_{\omega}\pi(a\vert\omega)\mu(\omega)=\alpha(a)$.
\end{enumerate}
\end{lem}
Lemma \ref{Lemma: straightforward rat} establishes that both of our
own criteria are equivalent to rationalization via \emph{straightforward
signal} (\citet{kamenica2011bayesian}), and the simplicity of these
objects facilitates a parsimonious characterization of the identified
set via the geometry of the receiver's payoff vectors. Toward that
end, write $\Phi\subset\mathbb{R}^{n}$ for the free disposal convex
hull of the image of $\phi(A)$. For each action $a$, write
\begin{align*}
\mathcal{N}(a) & \equiv\{\eta:\Omega\to\mathbb{R}\vert\sum_{\omega}\eta(\omega)=1,\ensuremath{\eta\cdot\phi(a)\geq\eta\cdot x\text{ for all \ensuremath{x\in\Phi}}\},}\\
\mathcal{N^{>}}(a) & \equiv\{\eta:\Omega\to\mathbb{R}\vert\sum_{\omega}\eta(\omega)=1,\ensuremath{\eta\cdot\phi(a)>\eta\cdot x\text{ for all \ensuremath{x\in\Phi}}\}}
\end{align*}
for the exterior unit normal vectors that delineate the supporting
hyperplanes and strictly supporting hyperplanes to $\Phi$ at $\phi(a)$,
respectively, noting that the latter are a subset of the former. 

We make two comments about these sets. First, because the payoff set
$\Phi$ is unbounded below, the components of $\eta$ are nonnegative
for each normal vector $\eta$ in the unit cone $\mathcal{N}(a)$.
From that point of view, $\mathcal{N}$ is appropriately interpreted
as a correspondence from the set of actions $A$ into the set of probability
distributions $\Delta(\Omega)$. Second, with the first point in hand,
it is apparent that $\mathcal{N}(a)=\{\mu\in\Delta(\Omega)\vert a\in a^{*}(\mu)\}$.

\begin{prop}
\label{Proposition : Identification}Prior distribution $\mu$ rationalizes
action distribution $\alpha$ if and only if there exists a selection
$n$ from $\mathcal{N}$ such that $\mu(\cdot)=\sum_{a\in A}n(\cdot\vert a)\alpha(a).$
\end{prop}
In addition to whatever guidance it might provide on identifying rationalizing
priors, Proposition \ref{Proposition : Identification} is informative
about both when the data is rationalizable at all and when it is rationalizable
by multiple priors. In order to state those results, it will be useful
to adapt some standard language to our context. Action $a\in A$ is
\emph{dominated }if there exists an action distribution $\gamma\in\Delta(A)$
such that $\sum_{b\in A}u(b\vert\omega)\gamma(b)>u(a\vert\omega)$
for all states $\omega$. More generally, action $a\in A$ is \emph{redundant
}if there exists an action distribution $\gamma\in\Delta(A)$ distinct
from $\delta(a)$ such that $\sum_{b\in A}u(b\vert\omega)\gamma(b)\geq u(a\vert\omega)$
for all states $\omega$. Mathematically, action $a$ is redundant
if and only if $\phi(a)$ is not an extreme point of $\Phi$. Economically,
$a$ is redundant if it is weakly dominated or a mixture of other
actions.

\begin{lem}
\label{Lemma: rationalizability existence}First, $\mathcal{N}(a)$
is empty if and only if action $a$ is dominated. Second, $\mathcal{N^{>}}(a)$
is empty if and only if action $a$ is redundant.
\end{lem}
While the first part of Lemma \ref{Lemma: rationalizability existence}
falls immediately out of the supporting hyperplane theorem, the second
part makes use of the fact that extreme points in finitely generated
convex sets have strictly supporting hyperplanes. As a matter of convenience,
we assume for the remainder of the paper that none of the receiver's
actions are redundant.

\section{\label{Section: two by two}The two-state two-action benchmark problem}

The two-state two-action problem is a useful benchmark for both the
two-state many-action problem and the two-action many-state problem.
Label the set of actions $A\equiv\{0,1\}$ to satisfy $v(0)<v(1)$
and the set of states $\Omega\equiv\{l,h\}$ to satisfy $\Delta(l)\equiv u(1\vert l)-u(0\vert l)<0<u(1\vert h)-u(0\vert l)\equiv\Delta(h),$
noting that the strict inequalities are consistent with our exclusion
of redundant actions from the receiver's decision problem. Further,
write 
\[
c\equiv\frac{\vert\Delta(l)\vert}{\vert\Delta(l)\vert+\Delta(h)}
\]
 for the quantity that satisfies $(1-c)\Delta(l)+c\Delta(h)=0$ and
note that action $1$ is optimal for the receiver under posterior
belief $\mu^{s}$ if and only if $\mu^{s}(h)\geq c$. Unsurprisingly,
the robustly optimal signal is such that this incentive compatibility
constraint is exactly satisfied under the lowest prior in the identified
set.
\begin{thm}
\label{Theorem: positive 2 states}Suppose there are two states and
two actions. Prior $\mu^{*}$ with $\mu^{*}(l)\equiv1-c\alpha(1),\mu^{*}(h)\equiv c\alpha(1)$
and direct signal $(\pi^{*},S^{*})$ with
\begin{align*}
\pi^{*}(a\vert l) & \equiv\frac{1}{\mu^{*}(l)}\begin{cases}
1-\alpha(1)+\mu^{*}(h) & a=0\\
\alpha(1)-\mu^{*}(h) & a=1,
\end{cases} & \pi^{*}(a\vert h) & \equiv\begin{cases}
0 & a=0\\
1 & a=1
\end{cases}
\end{align*}
are a saddle point of the sender's problem with value $\alpha$.
\end{thm}
The proof of Theorem \ref{Theorem: positive 2 states} is a formalization
of our discussion of the motivating example. We direct readers interested
in intuition back to the introduction and readers in want of a detailed
argument forward to the Appendix.

\section{\label{Section: two state many action}The two-state problem}

Label the set of states $\Omega\equiv\{l,h\}$ and the set of actions
$A\equiv\{0,...,k\}$. Where convenient, we identify state distributions
$\mu$ with the probability $\mu(h)$ they assign to the high state
$h$.
\begin{assumption}
\label{Assumption monotone}The sender's utility function satisfies
$v(0)<...<v(k)$ and the receiver's utility satisfies $u(0\vert h)<...<u(k\vert h)$. 
\end{assumption}
The rationalizing priors in the two state problem have a closed form
characterization. In light of our exclusion of redundant actions,
the receiver's expected utility functional satisfies a suitable form
of quasiconcavity and the bounds $\mathcal{N}(i)=[\xi(i),\xi(i+1)]$
with
\begin{align}
\xi(0) & \equiv0, & \xi(i) & \equiv\frac{u(i-1\vert l)-u(i\vert l)}{u(i-1\vert l)-u(i\vert l)+u(i\vert h)-u(i-1\vert h)}, & \xi(k+1) & \equiv1\label{display: nalow}
\end{align}
fall immediately out of payoff comparisons between adjacent actions.
In turn, Proposition \ref{Proposition : Identification} yields the
bound
\begin{equation}
\text{\ensuremath{\mu\in\mathcal{P}(\alpha)\iff\underline{\mu}\equiv\sum_{i}\xi(i)\alpha(i)\leq\mu\leq\sum_{i}\xi(i+1)\alpha(i)}\;\ensuremath{\equiv}\;\ensuremath{\overline{\mu}.}}\label{Display: two states many action identified set}
\end{equation}
Define value function $\widehat{v}:[0,1]\to\mathbb{R}$ and its concave
envelope $\widehat{V}:[0,1]\to\mathbb{R}$ by 
\begin{align*}
\widehat{v}(\mu) & \equiv v(\widehat{a}(\mu)), & \widehat{V}(\mu) & \equiv\max\{y\vert(\mu,y)\in\text{convex hull}(\text{graph}(\widehat{v}))\},
\end{align*}
noting that $\widehat{v}$ is nondecreasing per Assumption \ref{Assumption monotone}
and $\widehat{V}$ is well defined because $\widehat{v}$ is upper
semicontinuous. Because $(\underline{\mu},\widehat{V}(\underline{\mu}))$
lies on the boundary of $\text{convex hull}(\text{graph}(\widehat{v}))$,
there exists a pair\footnote{The supporting hyperplane theorem improves Caratheodory's $n+1$ to
$n$ on the boundary.} of posteriors $p,q$ and a constant $\lambda$ such that 
\[
(\underline{\mu},\widehat{V}(\underline{\mu})))=(1-\lambda)(p,\widehat{v}(p))+\lambda(q,\widehat{v}(q)).
\]
Consequently, if signal $(\pi,S)$ induces posterior $p$ with probability
$1-\lambda$ and posterior $q$ with probability $\lambda$ in the
known-prior counterfactual with distribution $\underline{\mu}$ then
$V((\pi,S)\vert\underline{\mu})\geq V((\pi',S')\vert\underline{\mu})$
for all signals $(\pi',S')$. We direct readers interested in further
justification of this assertion to \citet*{kamenica2011bayesian}
Section II.B.
\begin{thm}
\label{Theorem: 2 states 3 actions}Let Assumption \ref{Assumption monotone}
hold, suppose there are two states and three or more actions, and
let $p,q,\lambda$ satisfy $(\underline{\mu},\widehat{V}(\underline{\mu}))=(1-\lambda)(p,\widehat{v}(p))+\lambda(q,\widehat{v}(q))$.
Prior $\mu^{*}$ with $\mu^{*}(l)\equiv1-\underline{\mu},\mu^{*}(h)\equiv\underline{\mu}$
and direct signal $(\pi^{*},S^{*})$ with
\begin{align*}
\pi^{*}(a\vert l) & \equiv\frac{1}{\mu^{*}(l)}\begin{cases}
\mu^{*}(l)-(1-q)\lambda & a=\widehat{a}(p)\\
(1-q)\lambda & a=\widehat{a}(q),
\end{cases} & \pi^{*}(a\vert h) & \equiv\frac{1}{\mu^{*}(h)}\begin{cases}
\mu^{*}(h)-q\lambda & a=\widehat{a}(p)\\
q\lambda & a=\widehat{a}(q)
\end{cases}
\end{align*}
are a saddle point of the sender's problem with value $\widehat{V}(\underline{\mu})$.
\end{thm}
The proof of Theorem \ref{Theorem: 2 states 3 actions} is a straightforward
generalization of Theorem \ref{Theorem: positive 2 states}, with
a few additional steps. Unlike in the two-by-two problem, increasing
the prior probability of the high state might change the receiver's
best responses to any of the messages sent by the robustly optimal
signal $(\pi^{*},S^{*})$. However, because the sender's utility is
increasing in the receiver's action and the receiver's action is increasing
in the posterior probability of the high state, any such changes in
behavior benefit the sender.
\begin{cor}
\label{Corollary 2 state many action}Let Assumption \ref{Assumption monotone}
hold and suppose there are two states and three or more actions. The
value of the sender's problem $\widehat{V}(\underline{\mu})$ and
the value of the data $\sum_{i}v(i)\alpha(i)$ satisfy $V(\underline{\mu})\geq\sum_{i}v(i)\alpha(i)$,
with equality if and only if there exists a hyperplane $H$ such that
$(\xi(i),\widehat{V}(\xi(i))=(\xi(i),\widehat{v}(\xi(i))\in H$ for
all $i$ in the support of $\alpha$.
\end{cor}
Corollary \ref{Corollary 2 state many action} follows jointly from
Theorem \ref{Theorem: 2 states 3 actions} and the observation
\begin{align*}
\widehat{V}(\underline{\mu}) & =\widehat{V}(\sum_{i}\xi(i)\alpha(i))\geq\sum_{i}\widehat{V}(\xi(i))\alpha(i)\geq\sum_{i}\widehat{v}(\xi(i))\alpha(i)=\sum_{i}v(i)\alpha(i),
\end{align*}
where the equalities are definitional, the first inequality is Jensen's
inequality, and the second inequality is by $\widehat{V}\geq\widehat{v}$.
The condition for equality reduces to $\widehat{V}(\xi(i))=\widehat{v}(\xi(i))$
on the support of $\alpha$ if $\alpha$ has binary support. Otherwise,
if three or more actions appear in the data then violations of the
linearity criterion occur for generic specifications of sender and
receiver preferences, even if $\widehat{V}$ and $\widehat{v}$ coincide
at the appropriate points. We view the availability of guarantees
better than the data in typical specifications of our problem as a
positive result.
\begin{example}
\label{Example:" improvmeent}Suppose there are two states and three
actions $0,1,2$ with sender utility $v(i)=i$, receiver utility $\phi(0)=(1,-1)$,
$\phi(1)=(0,0)$, $\phi(2)=(-2,1)$, and action distribution $\alpha(i)=1/3$.
\end{example}
In Example \ref{Example:" improvmeent}, $\xi(1)=1/2,\xi(2)=2/3,$
prior $\mu^{*}\equiv\underline{\mu}=7/18$ is adversarial and signal
$(\pi^{*},S^{*})$ with $S^{*}\equiv\{0,2\}$, $\pi^{*}(2\vert l)\equiv7/22,\pi^{*}(2\vert h)\equiv1$
is robustly optimal. The sender's payoff guarantee $V((\pi^{*},S^{*})\vert\mu^{*})=\widehat{V}(\mu^{*})=7/6$
exceeds the expected value of the data $\sum_{a}v(a)\alpha(a)=1$. 

\section{\label{Section: two actions and many states}The two-action problem}

Label the set of actions $A\equiv\{0,1\}$ to satisfy $1\equiv v(1)>v(0)\equiv0$
and identify action distributions $\alpha$ with the probability $\alpha(1)$
they assign to the sender's preferred action. Assume the map $\Delta:\Omega\to\mathbb{R}$
defined by $\Delta(\omega)\equiv u(1\vert\omega)-u(0\vert\omega)$
is injective, label the states $\Omega\equiv\{\omega_{1},...,\omega_{n}\}$
to satisfy $\Delta(\omega_{1})<...<\Delta(\Omega_{n})$, and recall
our exclusion of redundant actions implies $\Delta(\omega_{1})<0<\Delta(\omega_{n})$.
For each state $\omega_{i}$, state distribution $\mu$, and signal
$(\pi,S)$ write $\Delta_{i}\equiv\Delta(\omega_{i})$, $\mu_{i}\equiv\mu(\omega_{i})$,
and $\pi_{i}\equiv\pi(\cdot\vert\omega_{i})$. Finally, write\emph{
}$L\equiv\{i\vert\Delta_{i}<0\}$ for the \emph{low} states, $H\equiv\{i\vert\Delta_{i}\geq0\}$
for the \emph{high} states, and for each pair of states $i\in L,j\in H$
write
\[
c_{ij}\equiv\frac{\vert\Delta_{i}\vert}{\vert\Delta_{i}\vert+\Delta_{j}}
\]
for the quantity that satisfies the indifference condition $(1-c_{ij})\Delta_{i}+c_{ij}\Delta_{j}=0$.
It will be helpful to keep in mind that $0\leq c_{ij}\leq1$, with
$c_{ij}<1$ except where $\Delta_{j}=0$.

Our analysis of the two-action problem is organized as follows. First,
in Section \ref{Section: rev princ} we (i) show the usual revelation
principle for persuasion problems does not hold and (ii) we compensate
for that failure by developing a complete characterization of optimal
signals for the known-prior counterfactual. Next, in Section \ref{Section: two state theorem 3}
we show that the sender's problem does not have a saddle point and
furthermore that there are no signals guaranteed to recover an action
distribution as good or better for the sender than the data. Finally,
with that negative result as motivation, we study robustly optimal
experimentation in Section \ref{Section: optimal experiments}.

\subsection{\label{Section: rev princ}The revelation principle}

Unlike in the two-state problem, the restriction to direct signals
is at least sometimes with loss in the two-action problem.
\begin{example}
\label{Example: rev}There are three states and two actions, the receiver's
utility satisfies $(\Delta_{1},\Delta_{2},\Delta_{3})=(-1,0,1)$,
and the data satisfy $\alpha=1/2$.

Consider Example \ref{Example: rev}. The identified set $\mathcal{P}(\alpha)$
is the set of priors $\mu$ satisfying the system $2\mu_{1}+\mu_{2}\geq1/2,\mu_{2}+2\mu_{3}\geq1/2.$
Suppose the sender uses direct signal $(\pi,S)$. If $\text{supp}(\pi_{1})$
and $\text{supp}(\pi_{2})$ intersect, then her payoff under prior
$(1/2,1/2,0)$ is $0$. Alternatively, if $\text{supp}(\pi_{1})$
and $\text{supp}(\pi_{2})$ are disjoint, then her payoff under the
prior $(3/4,0,1/4)$ is at most $1/4$. The signal with $\pi_{1}(1)\equiv0,\pi_{2}(1)\equiv\pi_{3}(1)\equiv1$
guarantees the sender payoff $1/4$ and is therefore robustly optimal
within the class of direct signals.
\end{example}
Suppose the sender instead uses the indirect signal $(\tilde{\pi},\tilde{S})$
with $\tilde{S}\equiv\{s_{1},s_{2},s_{3}\}$ and $\tilde{\pi}_{1}\equiv(5/6,0,1/6),\tilde{\pi}_{2}\equiv(0,1,0),\tilde{\pi}_{3}\equiv(0,0,1)$.
If $\mu_{1}\leq6\mu_{3}$ then the receiver's best response to messages
$s_{2}$ and $s_{3}$ is $1$, the sender's payoff is $(1/6)\mu_{1}+\mu_{2}+\mu_{3}$,
and the worst prior in this region $(3/4,0,1/4)$ yields payoff $3/8$.
Alternatively, if $\mu_{1}>6\mu_{3}$ then $1$ is a best response
to only $s_{2}$, the sender's payoff is $\mu_{2}$, and the limiting
worst-case prior $(6/10,3/10,1/10)$ in this region yields payoff
$3/10$. Accordingly, the indirect signal $(\tilde{\pi},\tilde{S})$
yields a better guarantee $3/10>1/4$ than the best direct signal. 
\begin{defn}
Signal $(\pi,S)$ is \emph{equivalent }to direct signal $(\pi',S')$
for prior $\mu$ if $\sum_{s\in S^{a}(\mu)}\pi_{i}(s)=\pi_{i}'(a)$
for all actions $a\in A$ and for all states $\omega_{i}\in\text{supp}(\mu)$.
\end{defn}
In light of Example \ref{Example: rev}, it will be useful to have
a complete characterization of the set of optimal signals for the
known-prior counterfactual. Toward that end, for each $i=1,...,n$
and each $z\in[0,1]$ define \emph{cutoff signal} $(\pi^{i,z},S^{i,z})$
by $S^{i,z}\equiv\{0,1\}$ and
\begin{align*}
\pi_{j}^{i,z}(0) & \equiv\begin{cases}
1 & j<i\\
1-z & j=i\\
0 & j>i,
\end{cases} & \pi_{j}^{i,z}(1) & \equiv\begin{cases}
0 & j<i\\
z & j=i\\
1 & j>i.
\end{cases}
\end{align*}
Prior $\mu$ is \emph{eligible if }$\sum_{i}\Delta_{i}\mu_{i}<0$.
If $\mu$ is indeed eligible, define parameters $(i(\mu),z(\mu))$
and value $\Pi(\mu)$ by
\begin{align*}
i(\mu) & \equiv\max\{i\vert\sum_{j\geq i}\Delta_{j}\mu_{j}<0\}, & z(\mu) & \equiv\frac{\sum_{j>i(\mu)}\Delta_{j}\mu_{j}}{\vert\Delta_{i(\mu)}\vert\mu_{i(\mu)}}, & \Pi(\mu)\equiv & \sum_{j>i(\mu)}(1+\frac{\Delta_{j}}{\vert\Delta_{i(\mu)}\vert})\mu_{j}.
\end{align*}
Otherwise, if $\mu$ is ineligible, define $(i(\mu),z(\mu),\Pi(\mu))\equiv(1,1,1)$.
In both cases, we use the shorthand $(\pi^{\mu},S^{\mu})\equiv(\pi^{i(\mu),z(\mu)},S^{i(\mu),z(\mu)})$.

We record a few facts about these signals here. First, the sender's
payoff in the known prior counterfactual satisfies $V((\pi^{\mu},S^{\mu})\vert\mu)=\Pi(\mu)$.
Second, the threshold state $i(\mu)$ belongs to the set of low states
$L$ for all priors $\mu$. Third, if $\mu$ is eligible, then $i(\mu)$
belongs to the support of $\mu$ and $z(\mu)<1$. Fourth, and finally,
in all cases signal $(\pi^{\mu},S^{\mu}$) is optimal for prior $\mu$
and that optimality is unique up to equivalence.
\begin{prop}
\label{Proposition: cutoff equivalence}Signal $(\pi^{*},S^{*})$
satisfies $\forall(\pi,S)\;V((\pi^{*},S^{*})\vert\mu)\geq V((\pi,S)\vert\mu)$
if and only if $(\pi^{*},S^{*})$ is equivalent to $(\pi^{\mu},S^{\mu})$.
\end{prop}
The characterization of optimality in Proposition \ref{Proposition: cutoff equivalence}
supports our results on the non-existence of saddle points later in
this section.

\subsection{\label{Section: two state theorem 3}Robustly optimal signals}

The basic logic of Theorem \ref{Theorem: positive 2 states} and Theorem
\ref{Theorem: 2 states 3 actions} --- in which the ``lowest''
prior is adversarial and its optimal signal is robustly optimal ---
fails with two actions and three or more states. 
\begin{example}
\label{example: not monotone iin FOSD}There are three states and
two actions, the receiver's utility satisfies $(\Delta_{1},\Delta_{2},\Delta_{3})=(-2,-1,1)$,
and the data satisfy $\alpha=1/2$. 
\end{example}
First, because $\mu\equiv(1/2,1/4,1/4)$ and prior $\nu\equiv(0,3/4,1/4)$
satisfy $\Pi(\mu)=\Pi(\nu)=\alpha$, both priors rationalize the data.
In turn, Proposition \ref{Proposition: cutoff equivalence} implies
cutoff signal $(\pi^{\mu},S^{\mu})$ with $(i(\mu),z(\mu))=(1,0)$
is optimal for the known-prior counterfactual in which states are
distributed according to $\mu$. However, even though $\nu$ first
order stochastically dominates $\mu$, the sender's disfavored action
$0$ is the receiver's unique best response to both messages sent
by $(\pi^{\mu},S^{\mu})$ under that prior. Accordingly, $V((\pi^{\mu},S^{\mu})\vert\nu)=0$
and the signal $(\pi^{\mu},S^{\mu})$ yields no payoff guarantee for
the sender. While there might be other signals that do somewhat better
for the sender, there are no signals that are guaranteed to an action
distribution at least as good as the data.
\begin{thm}
\label{Theorem: negative result three states}Suppose there are two
actions and three or more states. The sender's problem has a saddle
point if and only if $\alpha\in\{0,1\}$.
\end{thm}
The proof of Theorem \ref{Theorem: negative result three states}
extends the failure of the optimal signal $(\pi^{\mu},S^{\mu})$ in
our discussion of the example to the broader class of signals that
are equivalent to $(\pi,S)$. Our approach is to construct a pair
of priors $\mu,\nu$ with the property that (i) the value of the known-prior
counterfactual problem with prior $\mu$ and the value of the known-prior
counterfactual problem with prior $\nu$ are both $\alpha$ and (ii)
the set of solutions to those two problems do not intersect.
\begin{cor}
\label{Corollary: negative two by 3}Suppose there are two actions
and three or more states. There exists a signal $(\pi,S)$ with $\min_{\mu\in\mathcal{P}(\alpha)}V((\pi,S)\vert\mu)\geq\alpha$
if and only if $\alpha\in\{0,1\}$.
\end{cor}
Corollary \ref{Corollary: negative two by 3} follows logically from
Theorem \ref{Theorem: negative result three states} and the existence
of priors with $\Pi(\mu)=\alpha$. We direct readers interested in
the construction of such a prior to the theorem's proof.

\subsection{\label{Section: optimal experiments}Robustly optimal experiments}

The negative result in Theorem \ref{Theorem: negative result three states}
suggests that the sender might benefit from learning more about the
distribution of the variable, and that is indeed the case. In this
section of the paper, we characterize the additional information the
sender needs in order to construct a signal that is guaranteed to
work at least as well as the signal that generated the data. 

Formally, suppose that before choosing a signal $(\pi,S)$ the sender
simultaneously observes the action data $\alpha$ and the distribution
$\beta$ of output $\sigma$ generated by Blackwell \emph{experiment
}$A_{m\times n}$, where $A_{ti}\equiv\text{Prob}(\sigma=\sigma_{t}\vert\omega=\omega_{i})$,
$\beta_{t}\equiv\text{Prob}(\sigma=\sigma_{t})$, and $\mathcal{Q}(\beta)\equiv\{x\in\Delta(\Omega)\vert Ax=\beta\}.$\footnote{There are two technicalities to be addressed. First, we restrict attention
to outcome distributions $\beta$ with the property that the matrix
equation $Ax=\beta$ has solutions. We view this restriction as a
misspecification test. Second, every solution to $Ax=\beta$ necessarily
satisfies $\sum_{i}x_{i}=1$. Our approach to the non-negativity constraints
$x_{i}\geq0$ is to begin with fully supported distributions $\mu$
and work with perturbations of the form $\mu+d$ for vectors $d$
that are small in every component.} In keeping with our analysis of the sender's problem, our approach
is to identify experiments that are guaranteed to induce saddle points
in the \emph{experimenter's problem }
\[
\max_{(\pi,S)}\;\min_{\mu\in\mathcal{P}(\alpha)\cap\mathcal{Q}(\beta)}\;V((\pi,S)\vert\mu).
\]

\begin{defn}
Experiment $A$ is \emph{reliable }if for all pairs $(\alpha,\beta)$
there exists a signal $(\pi^{*},S^{*})$ and a prior $\mu^{*}\in\mathcal{P}(\alpha)\cap\mathcal{Q}(\beta)$
such that for all signals $(\pi,S)$ and for all priors $\mu\in\mathcal{P}(\alpha)\cap\mathcal{Q}(\beta)$
\begin{align*}
V((\pi^{*},S^{*})\vert\mu^{*}) & \geq V((\pi,S)\vert\mu^{*}), & V((\pi^{*},S^{*})\vert\mu^ {}) & \geq V((\pi^{*},S^{*})\vert\mu^{*}).
\end{align*}
If experiment $A$ is reliable, then the experimenter's problem has
a saddle point for every valid specification of the action data $\alpha$
and the experimental data $\beta$. Conversely, if $A$ is not reliable,
then there is a specification of the pair $(\alpha,\beta)$ under
which that is not the case.

We organize our analysis of optimal experiments into subsection. First,
we characterize the set of reliable\emph{ }experiments in Section
\ref{Section: reliable experiments}. Next, in Section \ref{Section: Simple experiments}
we specialize to \emph{simple} experiments that partition the state
space. Finally, in Section \ref{Section: sequential experiments}
we consider an alternative timing convention in which the sender chooses
the experiment after observing the action data.
\end{defn}

\subsubsection{\label{Section: reliable experiments}Reliable experiments}

The properties of the set $\mathcal{Q}(\cdot)$ of priors consistent
with the outcome of the sender's experiment are characterized by the
properties of the null space $\mathcal{D}\equiv\{x\vert Ax=0\}$ of
the stochastic matrix $A$. While the distribution of the state of
the world is exactly identified if $\mathcal{D}=\{0\}$, we have more
generally that prior $\nu$ is consistent with experimental data $A\mu$
if and only if $(\nu-\mu)\in\mathcal{D}$. We make repeated use of
the fact that 
\[
\sum_{i}d_{i}=\sum_{i}\sum_{t}A_{ti}d_{i}=\sum_{t}\sum_{i}A_{ti}d_{i}=0
\]
for all $d\in\mathcal{D}$ and further normalize nonzero \emph{directions}
in $\mathcal{D}$ to satisfy
\begin{align*}
\sum_{i\in H}d_{i} & \geq0, & \sum_{i\in H}d_{i}=0 & \implies\sum_{i\in H}\Delta_{i}d_{i}\geq0, & \sum_{i\in H}d_{i}=\sum_{i\in H}\Delta_{i}d_{i}=0 & \implies d_{\max\{i\vert d_{i}\neq0\}}>0.
\end{align*}
Normalizations in hand, the set of reliable experiments has a simple
characterization. 
\begin{defn}
\label{defn: ordered}Experiment $A$ is \emph{ordered }if $\sum_{j\geq i}d_{j}\geq0$
and $\sum_{j\geq i}\Delta_{j}d_{j}\geq0$ for all directions $d\in\mathcal{D}$
and for all states $i=1,...,\max L+1$.
\end{defn}
Consider the first condition $\sum_{j\geq i}d_{j}\geq0$ for all $i=1,...,\max L+1$.
If $d$ satisfies that condition then for every cutoff signal $(\pi,S)$
movements in the positive direction $\mu+d$ increase the probability
with which $(\pi,S)$ recommends the sender's preferred action to
the receiver. Conversely, if $d$ fails the first criterion in Definition
\ref{defn: ordered} then there exists a cutoff signal $(\pi^{i,z},S^{i,z})$
such that movements in the direction $\mu+d$ reduce the probability
that the better message is sent. Accordingly, we call the first criterion
the \emph{frequency criterion. }

Consider instead the second condition $\sum_{j\geq i}\Delta_{j}d_{j}\geq0$
for all $i=1,...,\max L+1$. If $d$ satisfies that condition then
for every prior $\mu$ and every cutoff signal movements in the positive
direction $\mu+d$ relax the incentive compatibility constraint $\sum_{i}\Delta_{i}(\mu+d)^{1}(\omega_{i})\geq0$
for the posterior induced by message $1$. Conversely, if $d$ fails
the condition then there exists a cutoff signal under which movements
in the direction $\mu+d$ tighten that incentive compatibility constraint.
We call the second condition the \emph{incentive compatibility criterion. }

For the purposes of illustration, recall Example \ref{example: not monotone iin FOSD}.
There, movements in the direction $d\equiv\nu-\mu=(-1/2,1/2,0)$ satisfy
the frequency criterion and the optimal signal $(\pi^{\mu},S^{\mu})$
for prior $\mu$ recommends the sender's preferred action with higher
probability under prior $\mu+\varepsilon d$ than under prior $\mu$.
However, because $(\Delta_{1}d_{1},\Delta_{2}d_{2},\Delta_{3}d_{3})=(1,-1/2,0)$,
movements in direction $d$ do not satisfy the frequency criterion
and the receiver's best response to message $1$ is action $0$. Accordingly,
under the interpretation of the sender's problem as a special case
of the experimenter's problem in which $A$ has rank $1$ and the
null space $\mathcal{D}$ of $A$ is the entire set of directions
with $\sum_{i}d_{i}=0$, our order criterion fails.

More generally, if both the frequency criterion and the incentive
compatibility criterion are satisfied the identified set is one dimensional
in the appropriate sense and we recover the positive results of Theorem
\ref{Theorem: positive 2 states}. On the other hand, if one of the
criteria fails, then the problem remains multidimensional and the
sender is in a position to benefit from additional information about
the distribution of the state of the world.

\begin{thm}
\label{Theorem: reliable iff ordered}Suppose there are two actions
and three or more states. Experiment $A$ is reliable if and only
if $A$ is ordered. 
\end{thm}
We develop the proof of Theorem \ref{Theorem: reliable iff ordered}
in several steps, with the positive result first as a standalone proposition.
\begin{prop}
\label{Experiments positive}Suppose there are two actions and three
or more states. If $A$ is ordered then $A$ is reliable.
\end{prop}
The basic idea behind Proposition \ref{Experiments positive} is as
follows. Suppose $A$ is ordered and consider any specification of
the data $(\alpha,\beta)$. Let prior $\mu\in\mathcal{P}(\alpha)\cap\mathcal{Q}(\beta)$
minimize the value function $\Pi(\cdot)$ on the identified set and
let $(\pi^{\mu},S^{\mu})$ be the optimal signal for the counterfactual
in which the prior is known to be $\mu$. 

Consider first movements in the positive direction $\mu+d$. Because
$A$ is ordered, these movements (i) increase the probability the
high signal is sent per the frequency criterion and (ii) relax the
incentive compatibility constraint for that signal per the incentive
compatibility criterion. Accordingly, such movements improve the performance
of the signal $(\pi^{\mu},S^{\mu})$. Consider instead movements in
the negative direction $\mu-d$. As we show, these movements unambiguously
degrade the value $\Pi(\cdot)$ of the prior. Because we began by
assuming that $\mu$ minimizes that value on the identified set, $\mu-d$
is thus excluded from $\mathcal{P}(\alpha)\cap\mathcal{Q}(\beta)$.
Accordingly, $\mu$ is adversarial and $(\pi^{\mu},S^{\mu})$ is robustly
optimal.

\begin{prop}
\label{Proposition: Negative result}Suppose there are two actions
and three or more states. If $A$ is not ordered then there exists
a data set $(\alpha,\beta)$ such that $\min_{\mu\in\mathcal{P}(\alpha)\cap\mathcal{Q}(\beta)}V((\pi,S)\vert\mu)<\alpha$
for all signals $(\pi,S)$.
\end{prop}
The converse of Proposition \ref{Experiments positive} follows logically
from Proposition \ref{Proposition: Negative result} because the existence
of a signal guaranteeing payoff at least $\alpha$ is a necessary
condition for the experimenter's problem to have a saddle point. We
develop the result in parts. First, we show that if our order criterion
fails then there necessarily exists a direction $d$ under which there
is tension between the frequency criterion and the incentive compatibility
criterion.
\begin{lem}
\label{Lemma: opposite signs}Suppose there are two actions and three
or more states. A is not ordered if and only if there exists a state
$1<i\leq\max L+1$ and a direction $d\in\mathcal{D}$ such that either
$\sum_{j\geq i}d_{j}>0>\sum_{j\geq i}\Delta_{j}d_{j}$ or $\sum_{j\geq i}\Delta_{j}d_{j}>0>\sum_{j\geq i}d_{j}$.
\end{lem}
Second, we show that there exists a prior $\mu$ and corresponding
signal $(\pi^{\mu},S^{\mu})$ for which the tension between the frequency
criterion and incentive compatibility criterion is relevant. 
\begin{lem}
\label{Lemma: rationalizing P}Suppose there are two actions and three
or more states. If $i\in L$ and $z\in(0,1)$ then there exists a
prior $\mu$ with full support and $(i(\mu),z(\mu))=(i,z)$. 
\end{lem}
Third, and finally, we confirm that movements in either the positive
direction $+d$ or the negative direction $-d$ are consistent with
the action data.

\begin{lem}
\label{lemma: Rationalizing Q}Suppose there are two actions and three
or more states. If $\mu$ has full support and $z(\mu)\in(0,1)$ then
for each $d\in\mathcal{D}$ at least one of $\mu+\varepsilon d,\mu-\varepsilon d$
belongs to $\mathcal{P}(\Pi(\mu))$.
\end{lem}
Together, Lemmas \ref{Lemma: opposite signs}--\ref{Lemma: rationalizing P}
provide the substance of our proof of Proposition \ref{Proposition: Negative result}.
Theorem \ref{Theorem: reliable iff ordered} then follows immediately
from the proposition and its predecessor. 
\begin{cor}
\label{corollary reliable just}Suppose there are two actions and
three or more states. If experiment $A$ is ordered then for all data
sets $(\alpha,\beta)$ there exists a signal $(\pi,S)$ such that
$\min_{\mu\in\mathcal{P}(\alpha)\cap\mathcal{Q}(\beta)}V((\pi,S)\vert\mu)\geq\alpha.$
Conversely, if $A$ is not ordered then there exists a data set $(\alpha,\beta)$
such that $\min_{\mu\in\mathcal{P}(\alpha)\cap\mathcal{Q}(\beta)}\;V((\pi,S)\vert\mu)<\alpha$
for all signals $(\pi,S)$.
\end{cor}
Corollary \ref{corollary reliable just} justifies our reliability
criterion on payoff grounds by clarifying that the set of experiments
under which the value of the experimenter's problem with data $(\alpha,\beta)$
is at least $\alpha$ is precisely the set of ordered (and hence reliable)
experiments. While the weak inequality in the first part of the suggests
that we might refine our criterion by ranking reliable experiments
according to the guaranteed value of the information they reveal,
no such refinement is possible because the prior $\mu^{\alpha}$ with
$\mu_{1}^{\alpha}\equiv1-c_{1n}\alpha,\mu_{n}^{\alpha}\equiv c_{1n}\alpha$
has value $\Pi(\mu^{\alpha})=\alpha$. Accordingly, the value of the
experimenter's problem with data $(\alpha,A\mu)$ is at most $\alpha$
for every reliable experiment $A$.

\subsubsection{\label{Section: Simple experiments}Simple experiments}

The interpretation of our characterization of reliability is complicated
by the need to relate the physical description of the sender's experiment
$A$ to the properties of its null space $\mathcal{D}$. Aside from
whatever analytical burden this intermediate step might impose, it
renders our order criteria at least somewhat opaque. Toward a remedy,
one example of a natural and easy-to-interpret class of experiments
are those that pin down the probabilities of some collection of subsets
of the state space $\Omega$. In this section, we use Theorem \ref{Theorem: reliable iff ordered}
to determine whether or not they are reliable.
\begin{defn}
Experiment $A$ is \emph{simple }if there exists a partition $\{\Sigma_{1},...,\Sigma_{k}\}$
of $\Omega$ with the property that for each consistent outcome $\beta$
there exists a probability distribution $\mu^{\beta}\in\Delta(\{1,...,k\})$
such that $\mu\in\mathcal{Q}(\beta)\iff(\mu(\Sigma_{1}),...,\mu(\Sigma_{k}))=(\mu_{1}^{\beta},...,\mu_{k}^{\beta}).$
\end{defn}
If the columns of $A$ are the canonical basis for $\mathbb{R}^{k}$
then $A$ is simple. More broadly, $A$ is simple if and only if its
columns modulo identity are linearly independent.
\begin{prop}
\label{Proposition: simple character}Experiment $A_{m\times n}$
is simple if and only if there exists a partition $\{\Sigma_{1},...,\Sigma_{k}\}$
of $\Omega$ and a matrix $B_{m\times k}$ such that (i) $A_{i}=B_{t}$
for all $i\in\Sigma_{t}$ and (ii) $B$ has rank $k$.
\end{prop}
In order to state our main result for this section, it will be useful
to introduce some terminology. If simple experiment $A$ induces partition
$\{\Sigma_{1},...,\Sigma_{k}\}$ then we say (i) $A$ \emph{separates
}state $i$ if there exists an index $t$ with $\Sigma_{t}=\{i\}$
and (ii) $A$ \emph{pools }states $E\subset\Omega$ if there exists
an index $t$ with $\Sigma_{t}=E$. 
\begin{thm}
\label{Theorem: which simple experiments are ordered}Suppose there
are two actions and three or more states. Simple experiment $A$ is
reliable if and only if (i) $A$ separates every low state $i$ or
(ii) $A$ pools one low state $i$ with one high state $j$ and separates
all other states.
\end{thm}
Theorem \ref{Theorem: which simple experiments are ordered} follows
jointly from Theorem \ref{Theorem: reliable iff ordered} and a straightforward
analysis of the structure of the null space $\mathcal{D}$ associated
with simple experiments. By way of interpretation, experiments $A$
that separate each of the negative states are reliable because the
prior $\mu$ that minimizes the sum $\sum_{i\in H}\Delta_{i}\mu_{i}$
over the identified set $\mathcal{P}(\alpha)\cap\mathcal{Q}(\beta)$
is apparently adversarial. Otherwise, if $A$ pools one low state
$i$ with one high state $j$ and separates all other states, then
the prior that maximizes the probability of $i$ over $\mathcal{P}(\alpha)\cap\mathcal{Q}(\beta)$
is itself adversarial.

\subsubsection{\label{Section: sequential experiments}Sequential experiments}

In formulating the experimenter's problem, we have implicitly assumed
that the sender must decide on an experiment $A$ before she observes
the action data $\alpha$. While this timing assumption strengthens
the sufficiency result in Proposition \ref{Experiments positive},
it also weakens the necessity result in Proposition \ref{Proposition: Negative result}.
In this section, we show that our results persist if the sender instead
observes the action data before she settles on an experiment.
\begin{defn}
Given action distribution $\alpha$, experiment $A$ is \emph{sequentially
reliable }if for all outcomes $\beta$ there exists a signal $(\pi^{*},S^{*})$
and a prior $\mu^{*}\in\mathcal{P}(\alpha)\cap\mathcal{Q}(\beta)$
such that for all signals $(\pi,S)$ and all priors $\mu\in\mathcal{P}(\alpha)\cap\mathcal{Q}(\beta)$
\begin{align*}
V((\pi^{*},S^{*})\vert\mu^{*}) & \geq V((\pi,S)\vert\mu^{*}), & V((\pi^{*},S^{*})\vert\mu) & \geq V((\pi^{*},S^{*})\vert\mu^{*}).
\end{align*}
\end{defn}
Our order criteria characterize sequential reliability under two additional
assumptions. First, we assume for expositional convenience that the
action distribution is interior.
\begin{assumption}
\label{assumption: alpha 0,1}The action distribution $\alpha$ satisfies
$\alpha\in(0,1).$
\end{assumption}
Second, and substantively, we assume that the state space satisfies
a straightforward richness condition.
\begin{assumption}
\label{Assumption: 0 state}There exists a state $i$ with $\Delta_{i}=0$.
\end{assumption}
The existence of a state in which the receiver is indifferent between
his two actions broadens the set of rationalizations of the data considerably.
We imagine a decision problem in which the state variable $\omega$
is ordered with many potential values and the indifference state $\omega_{i}$
is a minimum acceptable threshold for a positive outcome. Examples
of such thresholds abound in binary decision problems, wherein we
see hard cutoffs in ``pass/fail'' assessments, 95\% thresholds for
``statistical significance'', minimum test scores, grades, or class
ranks in admissions decisions and hiring decisions, and so on.

\begin{lem}
\label{Lemma: full support (i,z) sequential}Suppose there are two
actions and three or more states and let Assumptions \ref{assumption: alpha 0,1}--\ref{Assumption: 0 state}
hold. If $i\in L$ and $z\in(0,1)$ then there exists a prior $\mu$
with full support, $(i(\mu),z(\mu))=(i,z)$, and $\Pi(\mu)=\alpha$. 
\end{lem}
Lemma \ref{Lemma: rationalizing P} is an analogue to our earlier
Lemma \ref{Lemma: rationalizing P}, here for fixed action distributions.
Our characterization of optimal experiments then obtains under the
same line of argument as before.
\begin{thm}
\label{Theorem: sequential esxperimentation}Suppose there are two
actions and three or more states and let Assumptions \ref{assumption: alpha 0,1}--\ref{Assumption: 0 state}
hold. Experiment $A$ is sequentially reliable if and only if $A$
is ordered.
\end{thm}
In the Appendix, we briefly outline how Lemma \ref{Lemma: rationalizing P}
facilitates an adaptation of the proof of Theorem \ref{Theorem: reliable iff ordered}
to the sequential timing convention considered in this section.
\begin{cor}
\label{corollary: sequent reli}Suppose there are two actions and
three or more states and let Assumptions \ref{assumption: alpha 0,1}--\ref{Assumption: 0 state}
hold. If experiment $A$ is ordered then for all experimental outcomes
$\beta$ there exists a signal $(\pi,S)$ such that $\min_{\mu\in\mathcal{P}(\alpha)\cap\mathcal{Q}(\beta)}\;V((\pi,S)\vert\mu)\geq\alpha.$Conversely,
if experiment $A$ is not ordered then there exists an experimental
outcome $\beta$ such that $\min_{\mu\in\mathcal{P}(\alpha)\cap\mathcal{Q}(\beta)}\;V((\pi,S)\vert\mu)<\alpha$
for all signals $(\pi,S)$. 
\end{cor}
In Corollary \ref{corollary: sequent reli}, we provide an interpretation
of sequential reliability as a dynamic optimality criterion for the
version of our model in which the action data are revealed before
the sender chooses an experiment. The justification for its statements
are the same as the justifications for their analogues in Corollary
\ref{corollary reliable just}.

\section{\label{Section: conclusions}Conclusions}

This paper proposes a data-driven model of robust information provision
to a single receiver. We first characterize the set of priors that
are consistent with observed behavior and second identify robustly
optimal signals for a variety of specifications of the sender's problem.
In light of our restriction to models with either two actions or two
states, we leave open for future work the analysis of the many-state
many-action problem. While our expectation is that the results for
those models are more closely related to our analysis of the two-action
problem than the two-state problem, it is unclear whether or not our
characterization of optimal experiments might generalize in some useful
way to that setting.

\bibliographystyle{chicago}
\bibliography{ddp}

\appendix

\section{Appendix}

\subsection{Section \ref{SECTION: identification}}
\begin{proof}[Proof of Lemma \ref{Lemma: straightforward rat} ]
First, we show (ii) implies (iii). Second, we use the first part
to simplify our argument that (i) implies (ii). Because (iii) vacuously
implies (i), those two demonstrations are sufficient to establish
the claimed equivalence.

Proceeding with the first claim, suppose signal prior $\mu$, signal
$(\pi,S)$, and incentive compatible choice rule $C:S\to\Delta(A)$
are such that (i) $\sum_{\omega}\sum_{s}C(a\vert s)\pi(s\vert\omega)\mu(\omega)=\alpha(a)$
for all $a$ and (ii) $a\in a^{*}(\mu^{s})$ for all $a\in\text{supp}(C(s))$
for all $s$. Consider the signal$(\pi',S')$ with $S'\equiv\text{supp}(\alpha)$
and $\pi'(a\vert\omega)\equiv\sum_{s}C(a\vert s)\pi(s\vert\omega)$.
For each $a\in S'$ define $S^{a}\equiv\{s\in S\vert a\in\text{supp}(C(s))\}$
and for each $s\in S^{a}$ define
\[
\lambda(s\vert a)\equiv\frac{C(a\vert s)\sum_{\omega}\pi(s\vert\omega)\mu(\omega)}{\sum_{\omega}\sum_{s\in S^{a}}C(a\vert s)\pi(s\vert\omega)\mu(\omega)}.
\]
First, for each $a\in S'$ we have $\lambda(\cdot\vert a)\in\Delta(S^{a})$
and $\sum_{s\in S^{a}}\lambda(s\vert a)\mu^{s}=\mu^{a}$, where we
write $\mu^{s}$ for the posteriors induced by $(\pi,S)$ and $\mu^{a}$
for the posteriors induced by $(\pi',S')$. Second, for each $a\in S'$
the set $\{\mu\vert a\in a^{*}(\mu)\}$ is (i) convex and (ii) contains
posterior $\mu^{s}$ for each $s\in S^{a}$ by hypothesis. Together,
these two facts imply $\mu^{a}\in\{\mu\vert a\in a^{*}(\mu)\}$ and
thus $a\in a^{*}(\mu^{a})$, as required. Third, and finally, we have
as claimed $\sum_{\omega}\pi'(a\vert\omega)\mu(\omega)=\sum_{\omega}\sum_{s\in S}\varphi(a\vert s)\pi(s\vert\omega)\mu(\omega)=\alpha(a).$ 

Proceeding with the second claim, suppose $\mu$ rationalizes $\alpha$.
By definition, there exists an index $g=1,...,G$, distribution $\gamma\in\Delta(\{1,...,G\})$,
and action distributions $\alpha^{1},...,\alpha^{G}$ such that $\mu$
homogeneously rationalizes $\alpha^{g}$ for each $g$ and $\alpha=\sum_{g}\gamma^{g}\alpha^{g}$.
Because we have already shown (ii) implies (iii), for each $g$ there
exists a signal $(\pi^{g},S^{g})$ such that $S^{g}\subset A$, $a\in a^{*}((\mu^{g})^{a})$
for all $a\in S^{g}$, and $\sum_{\omega}\pi^{g}(a\vert\omega)\mu(\omega)=\alpha^{g}(a)$
for all $a$ in the support of $\alpha^{g}$, where we write $(\mu^{g})^{a}$
for the posterior induced by signal $(\pi^{g},S^{g})$ and message
$a\in S^{g}$. Define signal $(\pi,S)$ by $S\equiv\text{supp}(\alpha)$
and $\pi(a\vert\omega)\equiv\sum_{g}\gamma^{g}\pi^{g}(a\vert\omega)$.
For each $a\in S$ consider the convex combination $\lambda(\cdot\vert a)\in\Delta(\{1,...,G\})$
defined by 
\[
\lambda(g\vert a)\equiv\frac{\gamma^{g}\alpha^{g}(a)}{\alpha(a)}.
\]
The family of posteriors
\[
(\mu^{g})^{a}(\omega)\equiv\frac{\pi^{g}(s\vert\omega)\mu(\omega)}{\sum_{\omega'}\pi^{g}(s\vert\omega')\mu(\omega')}
\]
induced by signal $(\pi^{g},S^{g})$ and message $a$ and the posterior
\[
\mu^{a}(\omega)\equiv\frac{\pi(a\vert\omega)\mu(\omega)}{\sum_{\omega'}\pi^ {}(a\vert\omega')\mu(\omega')}
\]
induced by signal $(\pi,S)$ and message $a$ jointly satisfy $\mu^{a}=\sum_{g}\lambda(g\vert a)(\mu^{g})^{s}$
for each $a$. Because $a\in a^{*}((\mu^{g})^{a})$ by hypothesis
and the set $\{\mu\vert a\in a^{*}(\mu)\}$ is convex, we conclude
$a\in a^{*}(\mu^{a})$. In turn, because $\sum_{\omega}\pi(a\vert\omega)\mu(\omega)=\alpha(a)$,
$\mu$ homogeneously rationalizes $\alpha$ via signal $(\pi,S)$.
\end{proof}
\begin{proof}[Proof of Proposition \ref{Proposition : Identification}]
Suppose there exists a selection $n$ from $\mathcal{N}$ such that
$\mu=\sum_{a}n(\cdot\vert a)\alpha(a)$. Consider signal $(\pi,S)$
defined by $S\equiv\text{supp}(A)$ and 
\[
\pi(a\vert\omega)\equiv\frac{\alpha(a)n(\omega\vert a)}{\mu(\omega)},
\]
with $\pi(\cdot\vert\omega)$ assigned arbitrarily for states $\omega$
not in $\text{supp}(\mu)$. Direct calculation yields $\mu^{a}(\omega)=n(\omega\vert a),\sum_{\omega}\pi(a\vert\omega)\mu(\omega)=\alpha(a)$
and thus Lemma \ref{Lemma: straightforward rat} implies $\mu$ rationalizes
$\alpha$. 

Conversely, suppose $\mu$ rationalizes $\alpha$. Again per Lemma
\ref{Lemma: straightforward rat}, there exists a signal $(\pi,S)$
such that (i) $S\subset A$; (ii) $a\in a^{*}(\mu^{a})$ for all $a\in S$;
and (iii) $\sum_{\omega}\pi(a\vert\omega)\mu(\omega)=\alpha(a)$.
First, because $a\in a^{*}(\mu^{a})$ and $\mathcal{N}(a)=\{\mu\vert a\in a^{*}(\mu)\}$
we have $\mu^{a}\in\mathcal{N}(a)$. Accordingly, $a\mapsto\mu^{a}$
is a selection from $\mathcal{N}$. Second, Bayes rule implies
\[
\sum_{a}\mu^{a}(\omega)\alpha(a)=\sum_{a\in A}\Big(\frac{\pi(a\vert\omega)\mu(\omega)}{\sum_{\omega'}\pi(a\vert\omega')\mu(\omega')}\Big)\alpha(a)=\sum_{a\in A}\pi(a\vert\omega)\mu(\omega)=\mu(\omega)
\]
for each state $\omega$ and hence that $\mu=\sum_{a}\mu^{a}\alpha(a)$.
\end{proof}
\begin{proof}[Proof of Lemma \ref{Lemma: rationalizability existence}]
First, if action $a$ is dominated then payoff vector $\phi(a)$
belongs to the interior of $\Phi$ and thus $\mathcal{N}(a)$ is empty.
Conversely, if $a$ is not dominated then $\phi(a)$ belongs to the
boundary of $\Phi$ and the supporting hyperplane theorem implies
$\mathcal{N}(a)$ is nonempty.

Second, action $a$ is neither dominated nor redundant if and only
if payoff vector $\phi(a)$ is an extreme point of $\Phi$. We claim
more strongly that $\phi(a)$ exposed.\footnote{Point $\phi(a)$ is an \emph{extreme point }of $\Phi$ if $\phi(a)$
is not a convex combination of other points in $\Phi$. More specifically,
$\phi(a)$ is an \emph{exposed point }of $\Phi$ if there exists a
hyperplane that (i) supports $\Phi$ at $\phi(a)$ and (ii) intersects
$\Phi$ only at $\phi(a)$.} This fact follows readily from two well known results. First, because
$\Phi$ is a finitely generated convex set in the sense of \citet{rockafellar1970convex},
Theorem 19.1 of that text implies that $\Phi$ is the intersection
of finitely many closed half spaces. Second, because $\Phi$ is therefore
a polyhedron in the sense of \citet*{bertsimas1997introduction},
Theorem 2.3 of that text implies that $\phi(a)$ is indeed exposed.
By definition, then, exists a vector $\eta\in\mathcal{N^{>}}(a)$
such that $\eta\cdot\phi(a)>\eta\cdot x$ for all $x\in\Phi$.
\end{proof}

\subsection{Proof of Theorem \ref{Theorem: positive 2 states}}

First, it is apparent that $V((\pi^{*},S^{*})\vert\mu^{*})\geq V((\pi,S)\vert\mu^{*})$
for all signals $(\pi,S)$.\footnote{We direct the interested reader to Proposition \ref{Proposition: cutoff equivalence}
for a self-contained proof.} Second, Proposition \ref{Proposition : Identification} implies $\mu(h)\geq\mu^{*}(h)$
for all $\mu\in\mathcal{P}(\alpha)$. Third, and finally, we have
\begin{align*}
\mu(h)\geq\mu^{*}(h)\implies & \sum_{\omega}\pi^{*}(1\vert\omega)(\mu(\omega)-\mu^{*}(\omega)),\sum_{\omega}\Delta(\omega)(\mu^{1}(\omega)-(\mu^{*})^{1}(\omega))\geq0\\
\implies & V((\pi^{*},S^{*})\vert\mu)\geq V((\pi^{*},S^{*})\vert\mu^{*}).
\end{align*}
Accordingly, $V((\pi^{*},S^{*})\vert\mu)\geq V((\pi^{*},S^{*})\vert\mu^{*})=\alpha(1)$
for all rationalizing priors $\mu$.

\subsection{Proof of Theorem \ref{Theorem: 2 states 3 actions}}

Where convenient, we identify priors and posteriors with the probability
they assign to the high state $h$. Proceeding, per the discussion
in the body we have $V((\pi^{*},S^{*})\vert\mu^{*})\geq V((\pi,S)\vert\mu^{*})$
for all signals $(\pi,S)$. Let $\mu$ be any rationalizing prior.
We claim $V((\pi^{*},S^{*})\vert\mu)\geq V((\pi^{*},S^{*})\vert\mu^{*})$. 

First, (\ref{Display: two states many action identified set}) implies
$\mu\geq\mu^{*}$ and thus $\mu^{s}\geq(\mu^{*})^{s}$ for all messages
$s\in S^{*}$. In turn, because $\widehat{v}$ is nondecreasing, we
have $\widehat{v}(\mu^{s})\geq\widehat{v}((\mu^{*})^{s})$ for all
$s\in S^{*}$. Suppose without loss of generality $p\geq q$. Again
because $\widehat{v}$ is nondecreasing, we have $\widehat{v}(p)\geq\widehat{v}(q)$.
Furthermore, because $p\geq q\implies p\geq\mu^{*}\geq q$, direct
calculation yields $\pi^{*}(\widehat{a}(p)\vert h)\geq\pi^{*}(\widehat{a}(p)\vert l)$.
Accordingly, $\sum_{\omega}\pi^{*}(\widehat{a}(p)\vert\omega)\mu(\omega)\geq\sum_{\omega}\pi^{*}(\widehat{a}(p)\vert\omega)\mu^{*}(\omega).$
Altogether we have
\begin{align*}
V((\pi^{*},S^{*})\vert\mu) & =\sum_{\omega}\pi^{*}(\widehat{a}(p)\vert\omega)\mu(\omega)\widehat{v}(\mu^{\hat{a}(p)}))+\sum_{\omega}\pi^{*}(\widehat{a}(q)\vert\omega)\mu(\omega)\widehat{v}(\mu^{\hat{a}(q)}))\\
 & \geq\sum_{\omega}\pi^{*}(\widehat{a}(p)\vert\omega)\mu(\omega)\widehat{v}(p)+\sum_{\omega}\pi^{*}(\widehat{a}(q)\vert\omega)\mu(\omega)\widehat{v}(q)\\
 & =\sum_{\omega}\pi^{*}(\widehat{a}(p)\vert\omega)\mu(\omega)(\widehat{v}(p)-\widehat{v}(q))+\widehat{v}(q)\\
 & \geq\sum_{\omega}\pi^{*}(\widehat{a}(p)\vert\omega)\mu^{*}(\omega)(\widehat{v}(p)-\widehat{v}(q))+\widehat{v}(q)=V((\pi^{*},S^{*})\vert\mu^{*}),
\end{align*}
where the first inequality follows jointly from (i) $\widehat{v}(\mu^{s})\geq\widehat{v}((\mu^{*})^{s})$
for all $s\in S^{*}$, (ii) $(\mu^{*})^{\hat{a}(p)}=p,$ and (iii)
$(\mu^{*})^{\hat{a}(q)}=q$; and the second follows jointly from $\widehat{v}(p)\geq\widehat{v}(q)$
and $\sum_{\omega}\pi^{*}(\widehat{a}(p)\vert\omega)\mu(\omega)\geq\sum_{\omega}\pi^{*}(\widehat{a}(p)\vert\omega)\mu^{*}(\omega).$
The case with $p<q$ is symmetric.

\subsection{Proof of Proposition \ref{Proposition: cutoff equivalence}}

Let $\mu$ be any prior and $(\pi,S)$ be any signal. First, if $\mu$
is ineligible then apparently $(\pi,S)$ is optimal if and only if
$(\pi,S)$ is equivalent to $(i(\mu),z(\mu))=(1,1)$. Second, if $\mu$
is eligible and $\text{supp}(\mu)\cap H$ is empty, then every signal
is both vacuously optimal and equivalent to $(\pi^{\mu},S^{\mu})$
with $(i(\mu),z(\mu))=(\max\text{supp}(\mu),0)$. Third, if $\mu$
is eligible and $\text{supp}(\mu)\cap H=\min H$ with $\Delta_{\min H}=0$,
then $(\pi,S)$ is optimal if and only if $(\pi,S)$ is equivalent
to $(\pi^{\mu},S^{\mu})$ with $(i(\mu),z(\mu))=(\max(\text{supp}(\mu)\cap L),0)$.
Suppose to the contrary $\sum_{i}\Delta_{i}\mu_{i}<0<\sum_{i\in H}\Delta_{i}\mu_{i}$
and define
\begin{align*}
\nu^{0} & \equiv\frac{1}{1-V((\pi,S)\vert\mu)}\sum_{s\in S^{0}(\mu)}\sum_{i}\pi_{i}(s)\mu_{i}\mu^{s}, & \nu^{1} & \equiv\frac{1}{V((\pi,S)\vert\mu)}\sum_{s\in S^{1}(\mu)}\sum_{i}\pi_{i}(s)\mu_{i}\mu^{s}.
\end{align*}
There are two cases to consider. If $\max\text{supp}(\nu^{0})>\min\text{supp}(\nu^{1})$
then $(\pi,S)$ is apparently not equivalent to $(\pi^{\mu},S^{\mu})$.
We claim further that $(\pi,S)$ is not optimal for $\mu$. Define
\begin{align*}
i & \equiv\max\{j\vert\sum_{l\geq j}\mu_{l}\geq V((\pi,S)\vert\mu)\}, & z & \equiv\frac{V((\pi,S)\vert\mu)-\sum_{j>i}\mu_{j}}{\mu_{i}},
\end{align*}
consider the cutoff signal $(\pi^{i,z},S^{i,z})$, and write $\tilde{\nu}^{1}$
for the posterior induced by its signal $1$. By construction, $\tilde{\nu}^{1}$
is both distinct from $\nu^{1}$ and first order stochastically dominates
$\nu^{1}$. Because $\Delta_{i}$ is strictly increasing in $i$,
we therefore have $\sum_{i}\Delta_{i}\tilde{\nu}_{i}^{1}>\sum_{i}\Delta_{i}\nu_{i}^{1}\geq0$.
If $z<1$, then $V((\pi^{i,z+\varepsilon},S^{i,z+\varepsilon})\vert\mu)>V((\pi,S)\vert\mu)$.
Otherwise, if $z=1$, then $V((\pi^{j,\varepsilon},S^{j,\varepsilon})\vert\mu)>V((\pi,S)\vert\mu)$
for $j\equiv\max\{k\vert k<i,k\in\text{supp}(\mu)\}$. In either case,
$(\pi,S)$ is not an optimal signal for prior $\mu$.

Alternatively, if $\min\text{supp}(\nu^{1})\geq\max\text{supp}(\nu^{0})$
then there exists a cutoff signal $(\pi^{i,z},S^{i,z})$ that induces
posteriors $\nu^{0},\nu^{1}$ and hence satisfies $V((\pi,S)\vert\mu)=V((\pi^{i,z},S^{i,z})\vert\mu)$.
First, if $V((\pi^{i,z},S^{i,z})\vert\mu)>V((\pi^{\mu},S^{\mu})\vert\mu)$
then $\sum_{i}\Delta_{i}\nu_{i}^{1}<0$, in contradiction to the definition
of $S^{1}(\mu)$. Accordingly, $(\pi^{\mu},S^{\mu})$ is optimal for
prior $\mu$. Second, if $V((\pi^{i,z},S^{i,z})\vert\mu)=V((\pi^{\mu},S^{\mu})\vert\mu)$
then $(\pi^{\mu},S^{\mu})$ necessarily induces posteriors $\nu^{0},\nu^{1}$.
Consequently,
\begin{align*}
\nu_{j}^{1} & =\frac{1}{V((\pi,S)\vert\mu)}\sum_{s\in S^{1}(\mu)}\pi_{j}(s)\mu_{j}=\frac{1}{V((\pi^{\mu},S^{\mu})\vert\mu)}\pi_{j}^{\mu}(1)\mu_{j}
\end{align*}
for all $j$. Because $V((\pi,S)\vert\mu)=V((\pi^{i,z},S^{i,z})\vert\mu)=V((\pi^{\mu},S^{\mu})\vert\mu)$,
we have for all states $j$ $\sum_{s\in S^{1}(\mu)}\pi_{j}(s)\mu_{j}=\pi_{j}^{\mu}(1)\mu_{j}$.
We conclude that if $(\pi,S)$ is optimal for prior $\mu$ then $(\pi,S)$
is equivalent to $(\pi^{\mu},S^{\mu})$.

\subsection{Proof of Theorem \ref{Theorem: negative result three states}}

With regards to the two trivial cases, if $\alpha=0$ then $\mu\equiv\delta(1)\in\mathcal{P}(\alpha)$.
Apparently, $(i(\mu),z(\mu),\Pi(\mu))=(1,0,0)$ and $((\pi^{\mu},S^{\mu}),\mu)$
are a saddle point of the sender's problem. Otherwise, if $\alpha=1$
then no prior $\mu$ in the identified set $\mathcal{P}(\alpha)$
is eligible. In turn, we have $(i(\mu),z(\mu),\Pi(\mu))=(1,1,1)$
by definition and $((\pi^{\mu},S^{\mu}),\mu)$ are a saddle point
for every $\mu\in\mathcal{P}(\alpha)$. Alternatively, suppose $\alpha\in(0,1)$.
Our approach is to show that there does not exist a signal with guarantee
at least $\alpha$. There are two cases to consider. First, suppose
$\vert L\vert=1$ and consider prior $\mu^{\lambda}$ with $\mu_{1}^{\lambda}\equiv1-\lambda c_{12}\alpha-(1-\lambda)c_{13}\alpha,\mu_{2}^{\lambda}\equiv\lambda c_{12}\alpha,\mu_{3}^{\lambda}\equiv(1-\lambda)c_{13}\alpha$.
For all $\lambda$, direct calculation yields
\begin{align*}
i(\mu^{\lambda}) & =1, & z(\mu^{\lambda}) & =\frac{\alpha-\lambda c_{12}\alpha-(1-\lambda)c_{13}\alpha}{1-\lambda c_{12}\alpha-(1-\lambda)c_{13}\alpha}, & \Pi(\mu^{\lambda}) & =\alpha.
\end{align*}
In turn, Proposition \ref{Proposition: cutoff equivalence} implies
for all $\lambda\in(0,1)$
\[
V((\pi,S)\vert\mu^{\lambda})=\alpha\implies\sum_{s\in S^{1}(\mu^{\lambda})}\pi_{1}(s)=z(\mu^{\lambda}),\sum_{s\in S^{1}(\mu^{\lambda})}\pi_{2}(s)=\sum_{s\in S^{1}(\mu^{\lambda})}\pi_{3}(s)=1.
\]
The necessary conditions for $\pi_{2}$ and $\pi_{3}$ imply the set
of incentive compatible messages $S^{1}(\mu^{\lambda})=\text{supp}(\pi_{2})\cup\text{supp}(\pi_{3})$
for prior $\mu^{\lambda}$ does not vary with $\lambda$. On the other
hand, because the map $\lambda\mapsto z(\mu^{\lambda})$ is injective,
the necessary condition for $\pi_{1}$ implies the set $S^{1}(\mu^{\lambda})$
does vary with $\lambda$. Accordingly, any particular signal $(\pi,S)$
satisfies all three necessary conditions for at most one $\lambda$
in $(0,1)$.

Alternatively, suppose $\vert L\vert>1$ and consider prior $\mu$
with $\mu_{1}\equiv1-c_{1n}\alpha,\mu_{n}\equiv c_{1n}\alpha$ and
prior $\nu$ with $\nu_{1}\equiv1-\alpha,\nu_{2}\equiv\alpha-c_{2n}\alpha,\nu_{n}\equiv c_{2n}\alpha$.
It can quickly be verified
\begin{align*}
(i(\mu),z(\mu),\Pi(\mu)) & =(1,\frac{\alpha-c_{1n}\alpha}{1-c_{1n}\alpha},\alpha), & (i(\nu),z(\nu),\Pi(\nu)) & =(1,0,\alpha).
\end{align*}
In turn, Proposition \ref{Proposition: cutoff equivalence} yields
the apparently inconsistent conditions
\begin{align*}
V((\pi,S)\vert\mu)\geq\alpha\implies & \sum_{s\in S^{1}(\mu)}\pi_{1}(s)>0,\sum_{s\in S^{1}(\mu)}\pi_{n}(s)=1\implies\text{supp}(\pi_{1})\cap\text{supp}(\pi_{n})\neq\emptyset.\\
V((\pi,S)\vert\nu)\geq\alpha\implies & \sum_{s\in S^{1}(\nu)}\pi_{1}(s)=0,\sum_{s\in S^{1}(\nu)}\pi_{n}(s)=1\implies\text{supp}(\pi_{1})\cap\text{supp}(\pi_{n})=\emptyset.
\end{align*}

\subsection{Section \ref{Section: reliable experiments}}
\begin{proof}[Proof of Proposition \ref{Experiments positive}]
Suppose $A$ is ordered. Let $\alpha\in[0,1]$ and let $\beta$ be
such that $\mathcal{Q}(\beta)$ is nonempty. Let $\mu^{*}$ be a minimizer
for $\Pi(\mu)$ on $\mathcal{P}(\alpha)\cap\mathcal{Q}(\beta)$, noting
(i) $\mathcal{P}(\alpha)\cap\mathcal{Q}(\beta)$ is the intersection
of compact sets and therefore compact and (ii) $\Pi$ is continuous
by the maximum theorem. Write $(i^{*},z^{*})\equiv(i(\mu^{*}),z(\mu^{*})),(\pi^{*},S^{*})\equiv(\pi^{\mu^{*}},S^{\mu^{*}})$
and choose $d\in\mathcal{D}$.

Consider movements in the positive direction $\mu^{*}+d$. Because
$A$ is ordered, we have
\begin{align*}
z^{*}d_{i^{*}}+\sum_{j>i^{*}}d_{j} & \geq\min(\sum_{j\geq i^{*}}d_{j},\sum_{j>i^{*}}d_{j})\geq0,\\
\Delta_{i^{*}}z^{*}d_{i^{*}}+\sum_{j>i^{*}}\Delta_{j}d_{j} & \geq\min(\sum_{j\geq i^{*}}\Delta_{j}d_{j},\sum_{j>i^{*}}\Delta_{j}d_{j})\geq0.
\end{align*}
Accordingly, $V((\pi^{*},S^{*})\vert\mu^{*}+d)\geq V((\pi^{*},S^{*})\vert\mu^{*})$. 

Consider instead movements in the negative direction $\nu^{*}\equiv\mu^{*}-d$.
There are three cases to consider. First, if $\nu^{*}$ is ineligible
then
\[
\sum_{i}\Delta_{i}\mu^{*}=\sum_{i}\Delta_{i}\nu_{i}^{*}+\sum_{i}\Delta_{i}d_{i}\geq\sum_{i}\Delta_{i}d_{i}\geq0
\]
and thus $\mu^{*}$ is also ineligible. In turn, we have $(i(\mu^{*}),z(\mu^{*}),\Pi(\mu^{*}))=(i(\nu^{*}),z(\nu^{*}),\Pi(\nu^{*}))=(1,1,1)$
and hence $V((\pi^{*},S^{*})\vert\nu^{*})=V((\pi^{*},S^{*})\vert\mu^{*})=1$. 

Second, if $\nu^{*}$ is eligible and $\mu^{*}$ is not then $\Pi(\mu^{*})=1$
and $\Pi(\nu^{*})<1$. Because $\mu^{*}$ is a minimizer, we conclude
$\nu^{*}\notin\mathcal{P}(\alpha)$. 

Third, if $\nu^{*}$ and $\mu^{*}$ are both eligible then $\sum_{j\geq i^{*}}\Delta_{j}\nu^{*}\leq\sum_{j\geq i^{*}}\Delta_{j}\mu_{j}^{*}<0,$
where the first inequality follows from the fact $A$ is ordered and
the second from the definition of $i^{*}=i(\mu^{*})$. In turn, it
follows from the definition of $i(\cdot)$ that $i(\nu^{*})\geq i^{*}$. 

If $i(\nu^{*})=i^{*}$ then by the definition of $\Pi(\cdot)$
\[
\Pi(\mu^{*})-\Pi(\nu^{*})=\sum_{j>i^{*}}d_{j}+\frac{1}{\vert\Delta_{i^{*}}\vert}\sum_{j>i^{*}}\Delta_{j}d_{j}.
\]
Accordingly, if $\max(\sum_{j>i^{*}}d_{j},\sum_{j>i^{*}}\Delta_{j}d_{j})>0$
then $\Pi(\nu^{*})<\Pi(\mu^{*})$ and hence $\nu^{*}\notin\mathcal{P}(\alpha)$.
Alternatively, if $\sum_{j>i^{*}}d_{j}=\sum_{j>i^{*}}\Delta_{j}d_{j}=0$
then $d_{i^{*}}=0$ because $\Delta_{i^{*}}<0$ and $\sum_{j\geq i^{*}}d_{j},\sum_{j\geq i^{*}}\Delta_{j}d_{j}\geq0$.
In that case, we apparently have $V((\pi^{*},S^{*})\vert\nu^{*})=V((\pi^{*},S^{*})\vert\mu^{*})$. 

Otherwise, if $j^{*}\equiv i(\nu^{*})>i^{*}$ then
\begin{align*}
\Pi(\nu^{*}) & =\nu_{j^{*}}^{*}z(\nu^{*})+\sum_{j>j^{*}}\nu^{*}<\sum_{j\geq j^{*}}\nu_{j}^{*}\leq\sum_{j\geq j^{*}}\mu_{j}^{*}\leq\mu_{i^{*}}^{*}z^{*}+\sum_{j>i^{*}}\mu_{j}^{*}=\Pi(\mu^{*}),
\end{align*}
where the strict inequality follows from $\nu_{j^{*}}^{*}z(\nu^{*})<v_{j^{*}}^{*}$
via $j^{*}\in\text{supp}(\nu^{*})$ and $z(\nu^{*})<1$; the first
weak inequality follows from our hypothesis that $A$ is ordered;
and the second weak inequality follows from our hypothesis that $j^{*}>i^{*}$.
We again conclude that $\nu^{*}\notin\mathcal{P}(\alpha)$.
\end{proof}
\begin{proof}[Proof of Lemma \ref{Lemma: opposite signs}]
Suppose $A$ is not ordered. By definition, there exists a state
$1\leq i\leq L+1$ and a direction $d\in\mathcal{D}$ such that at
least one of the quantities $\sum_{j\geq i}d_{j},$ $\sum_{j\geq i}d_{j}\Delta_{j}$
is strictly negative. If $i=L+1$ then our normalizations $\mathcal{D}$
imply directly that $\sum_{j\in H}d_{j}>0>\sum_{j\in H}d_{j}\Delta_{j}$.
Alternatively, suppose $\sum_{j\in H}d_{j},\sum_{j\in H}d_{j}\Delta_{j}\geq0$
and let $i$ be the largest state in $L$ with the property that either
$\sum_{j\geq i}d_{j}<0$ or $\sum_{j\geq i}d_{j}\Delta_{j}<0$. Because
$\sum_{j>i}d_{j},\sum_{j>i}d_{j}\Delta_{j}\geq0$, we have
\begin{align*}
 & \sum_{j\geq i}d_{j}<0\leq\sum_{j>i}d_{j}\implies d_{i}<0\implies d_{i}\Delta_{i}>0\implies\sum_{j\geq i}d_{j}\Delta_{j}>\sum_{j>i}d_{j}\Delta_{j}\geq0,\\
 & \sum_{j\geq i}d_{j}\Delta_{j}<0\leq\sum_{j\geq i}d_{j}\Delta_{j}\implies d_{i}\Delta_{i}<0\implies d_{i}>0\implies\sum_{j\geq i}d_{j}>\sum_{j>i}d_{j}\geq0.
\end{align*}
If $i>1$, there is nothing more to show. Otherwise, if $i=1$ then
$\sum_{j\geq i}d_{j}=\sum_{j\geq1}d_{j}=0$ and hence $\sum_{j\geq i}d_{j}\Delta_{j}<0$
by hypothesis. Because
\[
\sum_{j\geq1}d_{j}\Delta_{j}<0\leq\sum_{j>1}d_{j}\Delta_{j}\implies d_{1}\Delta_{1}<0\implies d_{1}>0\implies\sum_{j\geq2}d_{j}<0,
\]
this contradicts our choice of $i$ to satisfy $\sum_{j>i}d_{j}\geq0$
and we conclude $i>1$. Because we have already established the claim
for that case, this completes the proof.
\end{proof}
\begin{proof}[Proof of Lemma \ref{Lemma: rationalizing P}]
Let $i\in L$ and $z\in(0,1)$. Define $\Delta^ {}\equiv\sum_{j>i}\Delta_{j}$
and let $\lambda\in[0,1)$ satisfy $(1-\lambda)\Delta^ {}+\lambda(n-1)\Delta_{n}>0$.
Further, define constants
\begin{align*}
\kappa & \equiv\frac{1}{\vert\Delta_{i}\vert}((1-\lambda)\Delta^ {}+\lambda(n-1)\Delta_{n}), & x & \equiv\frac{z}{(n-1)z+\kappa}
\end{align*}
and consider prior $\mu$ with $\mu_{i}\equiv1-(n-1)x,\mu_{n}\equiv(1-\lambda)x+\lambda(n-1)x,$
and $\mu_{j}\equiv(1-\lambda)x$ for all $j\neq i,n$. First, $\mu$
has full support because $\lambda<1$. Second, direct substitution
yields $\Delta_{i}z\mu_{i}+\sum_{j>i}\Delta_{j}\mu_{j}=\Delta_{i}\kappa x+\vert\Delta_{i}\vert\kappa x=0$
and thus $(i(\mu),z(\mu))=(i,z)$.
\end{proof}
\begin{proof}[Proof of Lemma \ref{lemma: Rationalizing Q}]
Suppose prior $\mu$ has full support. If $z(\mu)\in(0,1)$ then
by definition $\sum_{i\geq i(\mu)}\Delta_{i}\mu_{i}<0<\sum_{i>i(\mu)}\Delta_{i}\mu_{i}$
and thus $i(\mu+\varepsilon d)=i(\mu-\varepsilon d)=i(\mu)$. In turn,
\begin{align*}
\Pi(\mu+\varepsilon d) & =\sum_{i>i(\mu)}(\mu_{i}+\varepsilon d_{i})(1+\frac{\Delta_{i}}{\vert\Delta_{i(\mu)}\vert}), & \Pi(\mu-\varepsilon d) & =\sum_{i>i(\mu)}(\mu_{i}-\varepsilon d_{i})(1+\frac{\Delta_{i}}{\vert\Delta_{i(\mu)}\vert})
\end{align*}
and thus $\text{sign}(\Pi(\mu+\varepsilon d)-\Pi(\mu))=-\text{sign}(\Pi(\mu-\varepsilon d))$.
If $\Pi(\mu+\varepsilon d)=\Pi(\mu-\varepsilon d)=\Pi(\mu)$ then
the proof is complete. Alternatively, if $\Pi(\mu+\varepsilon d)>\Pi(\mu)$
then consider cutoff signal $(\pi^{i,z},S^{i,z})$ with $i\equiv i(\mu)$,
$z\equiv(\mu_{i}+\varepsilon d_{i})^{-1}(\Pi(\mu)-\sum_{j>i}(\mu_{j}+\varepsilon d_{j}))$
and note
\begin{align*}
\sum_{j}\pi_{j}^{i,z}(1)(\mu_{j}+\varepsilon d_{j}) & =\Pi(\mu), & \sum_{j}\Delta_{j}\pi_{j}^{i,z}(1)(\mu_{j}+\varepsilon d_{j}) & =\Delta_{j}(\Pi(\mu)-\Pi(\mu+\varepsilon d))>0,
\end{align*}
where the strict inequality follows from $i\in L$ and $\Pi(\mu+\varepsilon d)>\Pi(\mu)$.
Accordingly, $V((\pi^{i,z},S^{i,z})\vert\mu+\varepsilon d)=\Pi(\mu)$
and $\mu+\varepsilon d\in\mathcal{P}(\Pi(\mu))$. The case with $\Pi(\mu-\varepsilon d)>\Pi(\mu)$
is substantively identical.
\end{proof}
\begin{proof}[Proof of Proposition \ref{Proposition: Negative result}]
Suppose $A$ is not ordered. Per Lemma \ref{Lemma: opposite signs},
there exists a direction $d\in\mathcal{D}$ and a state $i$ with
$1<i<\max L+1$ such that one of
\begin{align}
\sum_{j\geq i}\Delta_{j}d_{j} & >0>\sum_{j\geq i}d_{j},\label{one direction}\\
\sum_{j\geq i}d_{j} & >0>\sum_{j\geq i}\Delta_{j}d_{j}\label{the other direction}
\end{align}
obtains. There are two substantive cases. First, suppose $i\in L$
and consider subcase (\ref{one direction}). Choose $z\in(0,1)$ to
satisfy $zd_{i}+\sum_{j>i}d_{j}<0<z\Delta_{i}d_{i}+\sum_{j>i}\Delta_{j}d_{i}$
and recall Lemma \ref{Lemma: rationalizing P} provides a prior $\mu$
with full support and $(i(\mu),z(\mu))=(i,z)$. In turn, Proposition
\ref{Proposition: cutoff equivalence} implies $(\pi,S)$ satisfies
$V((\pi,S)\vert\mu)\geq V((\pi',S')\vert\mu)$ for all signals $(\pi',S')$
if and only if $(\pi,S)$ is equivalent to $(\pi^{\mu},S^{\mu})$.
Consider the experimenter's problem with data $(\alpha,\beta)\equiv(\Pi(\mu),A\mu)$
and let $(\pi,S)$ be any such signal. Per Lemma \ref{lemma: Rationalizing Q},
either $\mu+\varepsilon d\in\mathcal{P}(\alpha)$ or $\mu-\varepsilon d\in\mathcal{P}(\alpha)$.
If $\mu+\varepsilon d\in\mathcal{P}(\alpha)$ then $S^{0}(\mu)\subset S^{0}(\mu+\varepsilon d)$
by $\sum_{i}\Delta_{i}(\mu+\varepsilon d)_{i}^{s}\approx\sum_{i}\Delta_{i}\mu_{i}^{s}<0$
for all messages $s\in S^{0}(\mu)$ and thus
\[
V((\pi,S)\vert\mu+\varepsilon d)=\sum_{j}\sum_{s\in S^{1}(\mu+\varepsilon d)}\pi_{j}(s)(\mu_{j}+\varepsilon d_{j})\leq\sum_{j}\sum_{s\in S^{1}(\mu)}\pi_{j}(s)(\mu_{j}+\varepsilon d_{j}).
\]
Because our choice of $z$ implies
\begin{align*}
\sum_{j}\sum_{s\in S^{1}(\mu)}\pi_{j}(s)(\mu_{j}+\varepsilon d_{j}) & =\alpha+\varepsilon(zd_{i}+\sum_{j>i}d_{j})<\alpha,
\end{align*}
we have $V((\pi,S)\vert\mu+\varepsilon d)<\alpha$. If instead $\mu-\varepsilon d\in\mathcal{P}(\alpha)$,
then we have $S^{0}(\mu)\subset S^{0}(\mu-\varepsilon d)$ by the
same argument as above and thus
\[
V((\pi,S)\vert\mu-\varepsilon d)=\sum_{j}\sum_{s\in S^{1}(\mu-\varepsilon d)}\pi_{j}(s)(\mu_{j}-\varepsilon d_{j})\leq\sum_{j}\sum_{s\in S^{1}(\mu)}\pi_{j}(s)(\mu_{j}-\varepsilon d_{j}).
\]
Because our choice of $z$ implies
\begin{align*}
\sum_{j}\sum_{s\in S^{1}(\mu)}\Delta_{j}\pi_{j}(s)(\mu_{j}-\varepsilon d_{j}) & =-\varepsilon(z\Delta_{i}d_{i}+\sum_{j>i}\Delta_{j}d_{j})<0,
\end{align*}
there exists a message $s^{*}\in S^{0}(\mu-\varepsilon d)\cap S^{1}(\mu)$
with $\sum_{j}\pi_{j}(s^{*})(\mu_{j}-\varepsilon d_{j})>0$. Together,
we obtain
\[
V((\pi,S)\vert\mu-\varepsilon d)\leq\sum_{j}\sum_{s\in S^{1}(\mu)\setminus\{s^{*}\}}\pi_{j}(s)(\mu_{j}-\varepsilon d_{j})<\sum_{j}\sum_{s\in S^{1}(\mu)}\pi_{j}(s)(\mu_{j}-\varepsilon d_{j})\approx\alpha.
\]
This establishes the argument for case $i\in L$ in subcase (\ref{one direction}).
Subcase (\ref{the other direction}) is symmetric. 

Second, suppose instead $i=\max L+1$ and consider subcase (\ref{one direction}).
Choose $z\in(0,1)$ to satisfy $zd_{i-1}+\sum_{j\geq i}d_{j}<0<zd_{i-1}\Delta_{i-1}+\sum_{j\geq i}d_{j}\Delta_{j}$
and recall again that Lemma \ref{Lemma: rationalizing P} provides
a cutoff signal $(i-1,z)$ and prior with full support $\mu$ such
that $(i(\mu),z(\mu))=(i-1,z)$. At this point, The argument is the
same as in the $i\in L$ case, with subcase (\ref{the other direction})
again symmetric to its counterpart.
\end{proof}

\subsection{Section \ref{Section: Simple experiments}}
\begin{proof}[Proof of Proposition \ref{Proposition: simple character}]
We write $d^{ij}$ for the vector in $\mathbb{R}^{n}$ satisfying
$d_{i}^{ij}\equiv1$ and $d_{j}^{ij}\equiv-1$ with $0$ in all other
entries. Suppose $A$ is simple. Let $\{\Sigma_{1},...,\Sigma_{k}\}$
be the associated partition of $\Omega$, let $\tau\equiv(\tau_{1},...,\tau_{k})$
be an arbitrary selection from the correspondence $t\rightrightarrows\Sigma_{t}$,
define experiment $B\equiv(B_{1},...,B_{k})$ by $B_{t}\equiv A_{\tau_{t}}$,
and let $\mu$ be any prior with full support. First, we claim $A_{i}=B_{t}$
for all $i\in\Sigma_{t}$. To see why, let $i,j\in\Sigma_{t}$ be
distinct and define $\nu\equiv\mu+\varepsilon d^{ij}$. Because $A$
is simple and $\nu(\Sigma_{t})=\mu(\Sigma_{t})$ for all $t$, we
have $\nu\in\mathcal{Q}(A\mu)$ and hence $(\nu-\mu)\in\mathcal{D}$.
Accordingly, $Ad^{ij}=A_{i}-A_{j}=0$, as claimed. Second, we claim
$B$ has full rank. To see why, suppose to the contrary that $Bx=0$
for some $(x_{1},...,x_{k})\neq(0,...,0)$ and let $(y_{1},...,y_{n})$
be any vector satisfying $\sum_{i\in\Sigma_{j}}y_{i}=x_{j}$ for all
indices $j$. Because $Ay=Bx=0$, we have $\nu\equiv(\mu+\varepsilon y)\in\mathcal{Q}(A\mu)$.
Further, because there exists index $t$ with $\sum_{i\in\Sigma_{t}}y_{i}=x_{t}\neq0$,
we have $\nu(\Sigma_{t})\neq\mu(\Sigma_{t})$. This contradicts our
hypothesis that $A$ is simple, and we conclude that $B$ indeed has
full rank.

Conversely, suppose there exists a partition $\{\Sigma_{1},...,\Sigma_{k}\}$
of $\Omega$ and an experiment $B\equiv(B_{1},...,B_{k})$ such that
(i) $A_{i}=B_{t}$ for all $i\in\Sigma_{t}$ and (ii) $B$ has full
rank. Let $\mu,\nu\in\Delta(\Omega)$, define $(x_{1},...,x_{k})\equiv(\mu(\Sigma_{1})-\nu(\Sigma_{1}),...,\mu(\Sigma_{k})-\nu(\Sigma_{k}))$,
and note that $\nu\in\mathcal{Q}(A\mu)\iff A(\mu-\nu)=Bx=0\iff x=0$.
Accordingly, $A$ is simple. 
\end{proof}
\begin{proof}[Proof of Theorem \ref{Theorem: which simple experiments are ordered}]
We again write $d^{ij}$ for the vector in $\mathbb{R}^{n}$ satisfying
$d_{i}^{ij}\equiv1$ and $d_{j}^{ij}\equiv-1$ with $0$ in all other
entries. Our approach is to show that $A$ is ordered if and only
if it satisfies one of the two criteria in the statement of the result.
Theorem \ref{Theorem: reliable iff ordered} then implies the result.

Suppose $A$ separates every low state and choose $d\in\mathcal{D}$.
Because $\sum_{i\in L}d_{i}=\sum_{i\in L}0=0$ by hypothesis and $\sum_{i}d_{i}=0$
by definition, we have $\sum_{i\in H}d_{i}=0$. By normalization,
then, we also have $\sum_{i\in H}\Delta_{i}d_{i}\geq0$. It follows
quickly from $d_{i}=0$ for all $i\in L$ that $\sum_{j\geq i}d_{i},\sum_{j\geq i}\Delta_{i}d_{i}\geq0$
for all $j=1,...,\max L+1$, and we conclude that $A$ is ordered.

Alternatively, suppose that $A$ pools exactly one pair of states
$(i,j)\in H\times L$ and separates all other states. It is immediate
that (i) $\mathcal{D}=\{\kappa d^{ij}\vert\kappa\geq0\}$ and (ii)
$\kappa d^{ij}$ satisfies both the frequency and the incentive compatibility
criteria for all $\kappa\geq0$. Accordingly, $A$ is ordered.

Conversely, suppose $A$ satisfies neither property. There are four
subcases. First, if there exist states $j<i\leq\max L$ with $i,j\in\Sigma_{t}$,
then $d^{ij}\in\mathcal{D}$. Because $\sum_{i'\geq i}\Delta_{i'}d_{i'}<0$,
$A$ is not ordered. Second, if there exist distinct states $i\leq\max L<j<k$
then $\mathcal{D}$ contains the direction $d\equiv d^{jk}+\varepsilon d^{ji}$.
Because $\sum_{\tau\geq j}\Delta_{\tau}d_{\tau}=(1+\varepsilon)\Delta_{j}-\Delta_{k}<0$,
$A$ is not ordered. Third, if there exist states $i<j\leq\max L<k<l$
and indices $t,t'$ such that $\Sigma_{t}=\{i,k\},\Sigma_{t'}=\{j,l\}$,
then $\mathcal{D}$ contains the direction $d\equiv d^{jl}+(1+\varepsilon)d^{ki}$.
Because $\sum_{\tau\in H}\Delta_{\tau}d_{\tau}=(1+\varepsilon)\Delta_{k}-\Delta_{l}<0$,
$A$ is not ordered. Fourth, and finally, if there exist states $i,j,k,l$
with $i\leq\max L<j,$ $\max L<k<l$ and indices $t,t'$ with $\Sigma_{t}=\{i,j\}$
and $k,l\in\Sigma_{t'}$ then $\mathcal{D}$ contains the direction
$d\equiv d^{sk}+\varepsilon d^{ji}$. Because $\sum_{\tau\in H}\Delta_{\tau}d_{\tau}=\varepsilon\Delta_{j}+(\Delta_{k}-\Delta_{l})<0$,
$A$ is not ordered.
\end{proof}

\subsection{Section \ref{Section: sequential experiments}}
\begin{proof}[Proof of Lemma \ref{Lemma: full support (i,z) sequential}]
Write $i^{0}$ for the state $i$ with $\Delta_{i^{0}}=0$ and recall
that Lemma \ref{Lemma: rationalizing P} provides a fully supported
prior $\mu$ with $(i(\mu),z(\mu))=(i,z)$, $\Delta_{i}\mu_{i}z+\sum_{j>i}\Delta_{j}\mu_{j}=0$.
Because $\mu$ has full support and $z>0$, we also have $\Pi(\mu)>0$.
If $\Pi(\mu)\leq\alpha$, define 
\[
\nu\equiv\frac{1-\alpha}{1-\Pi(\mu)}\mu+\frac{\alpha-\Pi(\mu)}{1-\Pi(\mu)}\delta(\omega_{i^{0}}).
\]
Alternatively, if $\Pi(\mu)>\alpha$ define 
\[
\nu\equiv\frac{\alpha}{\Pi(\mu)}\delta(\omega_{1})+\frac{\Pi(\mu)-\alpha}{\Pi(\mu)}\mu.
\]
In both cases, we have $\nu_{i}z+\sum_{j>i}\nu=\alpha$ and $\Delta_{i}\nu_{i}z+\sum_{j>i}\Delta_{j}\nu_{j}=0.$
Because $\nu$ inherits full support from $\mu$, it meets all the
criteria of the claim.
\end{proof}
\begin{proof}[Proof of Theorem \ref{Theorem: sequential esxperimentation}]
Fix action distribution $\alpha$. First, if $A$ is ordered then
Theorem \ref{Theorem: reliable iff ordered} implies $A$ is reliable.
Because reliability implies sequential reliability, $A$ is sequentially
reliable. Conversely, if $A$ is not ordered then the non-reliability
of $A$ follows from Lemma \ref{Lemma: opposite signs}, Lemma \ref{lemma: Rationalizing Q},
and Lemma \ref{Lemma: full support (i,z) sequential} from the same
argument given in the proof of Proposition \ref{Proposition: Negative result},
with Lemma \ref{Lemma: full support (i,z) sequential} as a surrogate
for Lemma \ref{Lemma: rationalizing P}.
\end{proof}

\end{document}